\pdfoutput=1
\documentclass[preprint,12pt]{elsarticle}

\usepackage{amsmath,amsfonts}
\usepackage{algorithmic}
\usepackage{algorithm}
\usepackage{array}
\usepackage[caption=false,font=normalsize,labelfont=sf,textfont=sf]{subfig}
\usepackage{textcomp}
\usepackage{stfloats}
\usepackage{url}
\usepackage{verbatim}
\usepackage{graphicx}

\usepackage{amssymb}
\usepackage{amsthm}

\usepackage{float}

\newtheorem{definition}{Definition}
\newtheorem{remark}{Remark}
\newtheorem{lemma}{Lemma}

\newtheorem{assumption}{Assumption}[section]
\newtheorem{theorem}{Theorem}

\usepackage{chngcntr}

\journal{Information Sciences}

\usepackage{geometry}   

\begin{document}

\begin{frontmatter}



\title{Non-trivial consensus on directed signed matrix-weighted networks with compound measurement noises and time-varying topologies}


\author{Tianmu Niu\fnref{label1}} 
\author{Xiaoqun Wu\corref{cor1}\fnref{label2}} 

            
\affiliation[label1]{organization={School of Mathematics and Statistics, Wuhan University},
	             city={Wuhan},
	             postcode={430072},
                 country={China}}
                 
\affiliation[label2]{organization={College of Computer Science and Software Engineering, Shenzhen University},
				city={Shenzhen},
				postcode={518060},
				country={China}}
				
\cortext[cor1]{Corresponding author.\\
E-mail addresses: tmniu.math@whu.edu.cn (T. Niu), xqwu@whu.edu.cn (X. Wu).}

\begin{abstract}
This paper studies non-trivial consensus\textemdash a relatively novel and unexplored convergence behavior\textemdash on directed signed matrix-weighted networks subject to both additive and multiplicative measurement noises under time-varying topologies. Building upon grounded matrix-weighted Laplacian properties, a stochastic dynamic model is established that simultaneously captures inter-dimensional cooperative and antagonistic interactions, compound measurement noises and time‑varying network structures. Based on stochastic differential equations theory, protocols that guarantee mean square and almost sure non-trivial consensus are proposed. Specifically, for any predetermined non-trivial consensus state, all agents are proven to converge toward this non-zero value in the mean-square and almost-sure senses. The design of control gain function in our protocols highlights a balanced consideration of the cumulative effect over time, the asymptotic decay property and the finite energy corresponding to measurement noises. Notably, the conditions on time-varying topologies in our protocols only require boundedness of elements in edge weight matrices, which facilitate the practicality of concept ``time-varying topology" in matrix-weighted network consensus algorithms. Furthermore, the proposed protocols operate under milder connectivity conditions and no requirements on structural (un)balance properties. The work in this paper demonstrates that groups with both cooperative and antagonistic inter-dimensional interactions can achieve consensus even in the presence of compound measurement noises and time-varying topologies, challenging the conventional belief that consensus is attainable only in fully cooperative settings.
\end{abstract}


\begin{highlights}
\item Achieves non-trivial consensus on matrix-weighted networks subject to compound noises, cooperative-antagonistic interactions and time-varying topologies.

\item Guarantees convergence of agent states to any preset non-zero consensus value in both the mean-square and almost-sure senses.

\item Imposes conditions on time-varying topologies that require only boundedness of elements in edge weight matrices.

\item Designs control gains that balance the cumulative effect, asymptotic decay, and finite energy of measurement noises.

\item Operates under relaxed connectivity conditions with no structural balance or unbalance requirements.

\end{highlights}


\end{frontmatter}



\section{Introduction}\label{Introduction}
Consensus, being a central topic in networked multi-agent systems (MASs), aims at enabling a group of agents to achieve a shared state through local interactions, typically on unsigned scalar-weighted graphs with fixed or switching topologies and negligible disturbances \cite{consensus problems} \cite{consensus_IS_2023}.

When communication measurements are corrupted by noises and disturbances, consensus dynamics become inherently stochastic. Based on stochastic differential equations theory \cite{SDE_Mao}, scholars have performed researches in depth on dynamic behaviors of MASs with noises. Works in \cite{noises_SIAM2018} investigated continuous-time multi-agent consensus with coexisting additive and multiplicative measurement noises, providing conditions for mean square and almost sure consensus under fixed and time-varying topologies. In particular, additive noise is typically treated as external disturbance independent of agent states, while multiplicative noise driven by relative state measurements introduces state-dependent uncertainty. Considering time-delays, \cite{noises_delay_Auto2019} developed consensus conditions for continuous-time MASs with additive measurement noises and for those with multiplicative measurement noises, respectively. For discrete-time Markov chains consensus protocols with additive noise, \cite{noises_Markov_TAC2019} studied their performance. Further, dynamic mean-square consensus for second-order hybrid MASs with time-varying delays \cite{MS_IF2023}, fault-tolerant consensus control under channel noises \cite{fault_noises_IF2023}, and event-triggered consensus for linear MASs with multiplicative noises \cite{noises_etc_Auto2026} are studied.

Beyond noises, antagonism naturally and commonly exists in reality. Therefore, signed networks—where edges may represent cooperative or antagonistic interactions—emerge and lead to qualitatively different collective behaviors. Classical dynamics on signed networks include (interval) bipartite or trivial consensus \cite{antagonistic} \cite{opinion separation} \cite{IB consensus} \cite{bipartite_IF2024}, where agents states might partition into two groups with opposite signs of the same magnitude. Altafini’s seminal work \cite{antagonistic} formalized such bipartite consensus dynamics and characterized structural balance conditions under which this behavior emerges. Mean square bipartite consensus problem for signed networks with communication noises is studied in \cite{noises_signed_2019}. On the other hand, by adding noises to the cooperative–competitive interactive information, \cite{noises_private_TAC2024} realized the differentially private bipartite consensus in mean-square and almost-sure senses over signed networks. 

However, the above consensus formulations predominantly consider scalar weights and structurally balanced topologies (for signed networks), with limitations to capture multi-dimensional coupling among agent states, for example, coupled oscillators and multiple-link pendulums dynamics \cite{small oscillations} \cite{LC}, the logical interdependence between different topics in opinion dynamics \cite{network science}	\cite{multiple interdependent topics} \cite{Luan}, and the graph effective resistances in distributed control and estimation \cite{Graph effective resistance} \cite{Barooah Estimation}. Matrix-weighted networks extend scalar-weighted interactions to reveal richer phenomena in consensus and clustering. In the study of consensus algorithms over matrix-weighted networks, a key geometric condition known as the ``positive spanning tree" was established for achieving consensus in undirected and unsigned networks by \cite{positive spanning tree}. For signed matrix-weighted networks, classical graph properties (e.g., structural balance) do not directly guarantee bipartite consensus, and new structures such as non-trivial balancing sets are required to characterize such steady states \cite{balancing set}. Subsequent researches on both directed and undirected signed networks have developed algebraic and geometric criterias for achieving bipartite and trivial consensus \cite{leader-following matrix-weighted consensus} \cite{Su TCASII matrix weighted bipartite consensus} \cite{Su TAC directed matrix weighted consensus}. The algebraic conditions are concentrated on the null space and spectral properties of matrix-weighted Laplacian \cite{Su TCASII matrix weighted bipartite consensus} \cite{Su TAC directed matrix weighted consensus} \cite{Pan TCASII bipartite consensus}. While prior researches \cite{Privacy Preservation_noises_matrix-weight_TCNS2025}	\cite{noises_mw_TCNS2020} have examined privacy preservation and consensus performance on matrix-weighted networks with noises, a gap remains in considering systems with antagonistic interactions and compound noises.

Non-trivial consensus, introduced in \cite{non-trivial consensus}, describes a relatively novel convergence behavior in signed networks, where agents reach agreement in both value and sign despite the coexistence of cooperative and antagonistic interactions. This behavior was previously deemed exclusive to fully cooperative groups. The non-trivial consensus paradigm holds notable theoretical and practical relevance, with potential applications ranging from formation control in UAV networks \cite{UAS} to opinion manipulation strategies within social networks \cite{Consensus manipulation} \cite{Manipulating opinions} \cite{manipulation Automatica}. In the latter context, the pervasive coexistence of trust and mistrust often leads to opinion separation among countries, political parties and opinionated individuals \cite{opinion separation}. Studying non-trivial consensus can thus inform strategies to steer a network containing mistrust relationships toward a shared opinion, thereby preventing final polarization or indifference on a given topic. Existing studies on non-trivial consensus, however, remain sparse. \cite{non-trivial consensus} established non-trivial consensus feasibility, albeit only for (essentially) cooperative networks, and \cite{Valcher} derived necessary and sufficient conditions for undirected signed networks—requiring both connectivity and structural balance without providing an explicit control design. \cite{SCIS Non-trivial consensus} realized non-trivial consensus on directed signed networks under both structurally balanced and unbalanced conditions. Nevertheless, all aforementioned works are confined to scalar-weighted network frameworks. One extension to matrix-weighted networks has been made in \cite{my_matrix-weighted_NTC}, which offered the conditions to gurantee that every eigenvalue of the grounded matrix-weighted Laplacians has positive real part, and further developed the systematic approach to achieve non-trivial consensus on directed signed matrix-weighted networks with both fixed and switching topologies. However, this approach operates under the ideal conditions without measurement noises, and requires a lower bound for the dwell time $\Delta t_{k}$. 

In view of the preceding discussion, it is still an open issue to investigate non-trivial consensus on signed matrix-weighted networks with compound noises and directed time-varying topologies. The primary contributions of this work, along with comparisons to some existing related achievements are summarized as follows.

\begin{enumerate}
\item Firstly, the model of directed signed matrix-weighted networks subject to compound noises is established. The inter-dimensional antagonistic interactions, the additive noises that are independent of agents states, the multiplicative noises driven by relative state measurements are jointly characterized in this model. Our model appropriately extends the classic SDEs model in \cite{noises_SIAM2018} to the more general matrix-weighted framework, which possesses greater generality and broader applicability. Building upon grounded matrix-weighted Laplacian properties in \cite{my_matrix-weighted_NTC}, the non-trivial consensus tasks are further converted into the stability problems of non-trivial consensus error.
	
\item Secondly, mean square non-trivial consensus and almost sure non-trivial consensus on signed matrix-weighted networks with directed fixed topologies and compound noises are realized. Specifically, Given a target non-trivial consensus state $\boldsymbol{\theta}\neq\boldsymbol{0}_{d}$, the convergence of agent states towards this shared non-zero state $\boldsymbol{\theta}$ in both the mean-square and the almost-sure senses, respectively, is guaranteed by the synthesis of informed agents $\mathcal{V}_{\mathcal{I}}(t)$ selection, coupling coefficients $\delta_{i}(t)$, coupling matrix weights $B_{i}(t)$ and the external signals $x_{0}(t)$ determination and the control gain $c(t)$ design. The design of $c(t)$ highlights a balanced consideration of the cumulative effect over time, the asymptotic decay property and the finite energy corresponding to measurement noises. Compared with most existing achievements on consensus algorithms of MASs with noises \cite{consensus_IS_2023} \cite{noises_SIAM2018} \cite{fault_noises_IF2023} \cite{noises_etc_Auto2026} \cite{noises_signed_2019} \cite{noises_mw_TCNS2020}, we realized non-trivial consensus, a relatively novel convergence result on  matrix-weighted networks that incorporate compound noises, antagonistic interactions and directed topologies. Meanwhile, unlike existing works on consensus algorithms of matrix-weighted networks that often demand positive spanning tree \cite{positive spanning tree} \cite{Su TAC directed matrix weighted consensus} or positive-negative spanning tree \cite{balancing set} \cite{Pan TCASII bipartite consensus} conditions, our protocols operate under milder connectivity conditions and no requirements on structural (un)balance properties.

\item Finally, mean square non-trivial consensus and almost sure non-trivial consensus on signed matrix-weighted networks with directed time-varying topologies and compound noises are further accomplished. Compared with existing works on consensus algorithms of matrix-weighted networks with time-varying topologies \cite{my_matrix-weighted_NTC} \cite{switching matrix-weighted consensus} \cite{switching matrix-weighted cluster consensus}, our topology Assumption serves as a relatively milder condition, since we do not theoretically assume the existence of dwell times, instead, only boundedness of elements in edge weight matrices is required. This condition allows the network topologies to evolve continuously over time, facilitates the practicality of concept ``time-varying topology" in matrix weighted network consensus algorithms.
\end{enumerate}

The remainder of this paper is organized as follows. Section \ref{Preliminaries} introduces notations and basic descriptions of general signed matrix-weighted networks with time-varying topologies and control gain function. Section \ref{Problem formulation} presents the specific network models with compound measurement noises, the standard definitions of mean square non-trivial consensus and almost sure non-trivial consensus in our study, along with a basic Lemma on the properties of grounded matrix-weighted Laplacian. Section \ref{Main results} presents the main results: mean square and almost sure non-trivial consensus are established for signed matrix-weighted networks with fixed topologies, and the same two convergence behaviors are further guaranteed for the case of time-varying topologies. Section \ref{Numerical examples} carries out the numerical simulations to verify our theoretical results. Finally, Section \ref{Conclusions} concludes this paper.

\section{Preliminaries}\label{Preliminaries}
\subsection{Notations}
Let $\mathbb{R},\ \mathbb{C},\ \mathbb{N},\ \mathbb{N}^{*}$ denote the set of real numbers, complex numbers, natural numbers and positive integers, respectively. For any $M\in\mathbb{N}^{*}$, $\underline{M}\triangleq\{1,...,M\}$. $\boldsymbol{0}_{d\times d}\in\mathbb{R}^{d\times d},\ \boldsymbol{0}_{d}\in\mathbb{R}^{d}$ represent the matrix and vector whose elements are all zero. The symmetric matrix $Q=Q^{\top}\in\mathbb{R}^{d\times d}$ is positive (semi-)definite, denoted as $Q\succ\boldsymbol{0}\ (Q\succeq\boldsymbol{0})$, if for any $\boldsymbol{\alpha}\in\mathbb{R}^{d},\ \boldsymbol{\alpha}\neq\boldsymbol{0}_{d}$, there holds  $\boldsymbol{\alpha}^{\top}Q\boldsymbol{\alpha}>0\ (\boldsymbol{\alpha}^{\top}Q\boldsymbol{\alpha}\geq0)$. While it is negative (semi-) definite, denoted as $Q\prec\boldsymbol{0}\ (Q\preceq\boldsymbol{0})$, if for any $\boldsymbol{\alpha}\in\mathbb{R}^{d},\ \boldsymbol{\alpha}\neq\boldsymbol{0}_{d}$, there holds  $\boldsymbol{\alpha}^{\top}Q\boldsymbol{\alpha}<0\ (\boldsymbol{\alpha}^{\top}Q\boldsymbol{\alpha}\leq0)$. More generally, for any two symmetric matrices $P,\ Q\in\mathbb{R}^{d\times d}$, denote $P\succ Q (P\succeq Q)$ if $P-Q\succ\boldsymbol{0} (P-Q\succeq\boldsymbol{0})$, and $P\prec Q (P\preceq Q)$ if $P-Q\prec\boldsymbol{0} (P-Q\preceq\boldsymbol{0})$.

The matrix-valued sign function $\mathbf{sgn}(\cdot)$ is employed to express the positive/negative (semi-) definiteness of a real symmetric matrix $Q$, to be specific,

\begin{equation}\notag
\mathbf{sgn}(Q)
\begin{cases}
1,\ Q\succeq\boldsymbol{0},\ Q\neq\boldsymbol{0}_{d\times d},\\
0,\qquad\qquad Q=\boldsymbol{0}_{d\times d},\\
-1,\ Q\preceq\boldsymbol{0},\ Q\neq\boldsymbol{0}_{d\times d},
\end{cases}
\end{equation}
and the function $|\cdot|$ for such symmetric matrix $Q$ is defined as $|Q|=\mathbf{sgn}(Q)\cdot Q$. The null space of a matrix $Q\in\mathbb{R}^{d\times d}$ is $\mathbf{null}(Q)=\{\boldsymbol{\alpha}\in\mathbb{R}^{d}|Q\boldsymbol{\alpha}=\boldsymbol{0}_{d}\}$. A block matrix $\mathcal{A}$ can be denoted as $\mathcal{A}=[A_{ij}]\in\mathbb{R}^{Nd\times Nd}$, with submatrix $A_{ij}\in\mathbb{R}^{d\times d},\ i,j\in\underline{N}$ being its $(i,j)$th block. Symbol $\otimes$ represents the Kronecker product.

\subsection{Graph Theory}
A signed matrix-weighted communication graph with $N$ vertices and time-varying topology is represented by a triple $\mathcal{G}(t)=(\mathcal{V},\mathcal{E}(t),\mathcal{A}(t))$, in which $\mathcal{V}=\{v_{1},...,v_{N}\},\ \mathcal{E}(t)\subseteq\mathcal{V}\times\mathcal{V}$ are the vertices set and the edges set, respectively. For directed $\mathcal{G}(t)$, $(v_{j},v_{i})\in\mathcal{E}(t)$ if and only if there exists a directed edge from vertex $v_{j}$ to $v_{i}$ at time $t$, and when $\mathcal{G}(t)$ is undirected, $(v_{j},v_{i})\in\mathcal{E}(t)$ and $(v_{i},v_{j})\in\mathcal{E}(t)$ holds simultaneously. $\mathcal{A}(t)=[A_{ij}(t)]\in\mathbb{R}^{Nd\times Nd}$ is the block adjacency matrix, in which the symmetric submatrix $A_{ij}(t)\in\mathbb{R}^{d\times d}$ represents the connectivity between vertices $v_{i}$ and $v_{j}$. If $\mathbf{sgn}(A_{ij}(t))=0$, then there exists no edge from $v_{j}$ to $v_{i}$ in $\mathcal{G}(t)$, otherwise, $\mathbf{sgn}(A_{ij}(t))=1$ if and only if $(v_{j},v_{i})\in\mathcal{E}(t)$ is a positive (semi-) definite edge, and $\mathbf{sgn}(A_{ij}(t))=-1$ if and only if $(v_{j},v_{i})\in\mathcal{E}(t)$ is a negative (semi-) definite edge.
When $\mathcal{G}(t)$ is undirected, $A_{ji}(t)=A_{ij}(t)$ for any $i,j\in\underline{N},\ i\neq j$. For any vertex $v_{i}\in\mathcal{V}$, its in-neighbor vertices set at time $t$ is defined as $\mathcal{N}_{i}(t)=\{v_{j}\in\mathcal{V}|\mathbf{sgn}(A_{ij}(t))\neq0\}$, and the out-neighbor vertices set is $\mathcal{N}'_{i}(t)=\{v_{j}\in\mathcal{V}|\mathbf{sgn}(A_{ji}(t))\neq0\}$. Define the vertices subset $\mathcal{U}(t)\subseteq\mathcal{V}$ as $\mathcal{U}(t)=\{v_{i}\in\mathcal{V}|\exists j\in\underline{N}\ s.t.\ \mathbf{sgn}(A_{ij}(t))=-1\}$. For any $v_{i}\in\mathcal{V}$, define $\varOmega_{i}(t)=\{v_{j}\in\mathcal{V}|\mathbf{sgn}(A_{ij}(t))=-1\}$, and $\varGamma_{i}(t)=\{v_{j}\in\mathcal{V}|\mathbf{sgn}(A_{ij}(t))=1\}$.

\subsection{Model Description}
Generally, consider the directed signed matrix-weighted networks with time-varying topology $\mathcal{G}(t)=(\mathcal{V},\mathcal{E}(t),\mathcal{A}(t))$:
\begin{equation}\label{FAN}
	\dot{x}_{i}(t)=c(t)\sum_{j\in\mathcal{N}_{i}(t)}\left|A_{i j}(t)\right|\left[\mathbf{sgn}\left(A_{ij}(t)\right) x_{j}(t)-x_{i}(t)\right],\ i\in\underline{N},
\end{equation}
where $x_{i}(t)\in\mathbb{R}^{d}$ is the state of agent $i$ at time $t$. The time-varying control gain $c(t): [0,+\infty)\mapsto[0,+\infty)$ is a continuous function. The matrix-weight $A_{ij}(t)\in\mathbb{R}^{d\times d}$ on edge $(i,j)$ is symmetric, i.e. $A_{ij}^{\top}(t)=A_{ij}(t)$. $A_{ij}(t)$ is either positive (semi)-definite, negative (semi)-definite, or a zero matrix.

The signed Laplacian $L(\mathcal{G}(t))=[L_{ij}(t)]\in\mathbb{R}^{Nd\times Nd}$ corresponding to the signed matrix-weighted graph $\mathcal{G}(t)$, or $L(t)$ for simplicity, is defined as \cite{balancing set}
\begin{equation}\notag
L_{ij}=
\begin{cases}
-A_{ij},\qquad\qquad\ j\neq i, \\
\sum_{k=1,k\neq i }^{N}\lvert A_{ik}\rvert,\ j=i,
\end{cases}
\end{equation}
in which $L_{ij}(t)\in\mathbb{R}^{d\times d}$ is the $(i,j)$th block of $L(t)$.

\section{Problem formulation}\label{Problem formulation}
In reality, systems are often disrupted by noises. For each agent, the information from its neighbors may contain different types of communication/measurement noises, which extends \eqref{FAN} to the following form:
\begin{equation}
\begin{aligned}\label{FAN_noise}
\dot{x}_{i}(t)=&c(t)\sum_{j\in\mathcal{N}_{i}(t)}\left|A_{i j}(t)\right|[\mathbf{sgn}\left(A_{ij}(t)\right) x_{j}(t)-x_{i}(t)+\boldsymbol{1}_{d}\sigma_{ji}\xi_{1ji}(t)\\
&+f_{ji}\left(x_{j}(t)-x_{i}(t)\right)\xi_{2ji}(t)],\ i\in\underline{N},
\end{aligned}
\end{equation}
where $\xi_{lji},\ l=1,2$ denote the measurement noises, $\sigma_{ji}>0$, and function $f_{ji}(\cdot):\mathbb{R}^{d}\mapsto\mathbb{R}^{d}$. We assume that the measurement noises are independent Gaussian white noises. i.e.
\begin{align}\label{Gaussian noises}
\int_{0}^{t}\xi_{lji}(s)ds=w_{lji}(t),\ t\geq0,\ i,j\in\underline{N},\ l=1,2,
\end{align}
where $\{w_{lji}(t),\ i,j\in\underline{N},\ l=1,2\}$ are scalar independent Brownian motions.

The noise intensity $f_{ji}(\cdot)$ needs to satisfy the following assumption.

\begin{assumption}\label{Lipschitz Assumption}
{\rm There exists a positive constant $\bar{\sigma}$ such that
\begin{equation}
\Vert f_{ji}(x)\Vert\leq\bar{\sigma}\Vert x\Vert,\ \forall i,j\in\underline{N}.
\end{equation}
}
\end{assumption}

The primary focus in this paper is to realize non-trivial consensus for system \eqref{FAN_noise}, i.e. driving these $N$ agents toward a non-zero consensus state, where the consensus state value can be arbitrarily configured as needed. Standard definitions of non-trivial consensus in our framework shall be given afterwards.

\begin{remark}
{\rm
The challenges brought by mechanism defined in \eqref{FAN_noise} lie in three aspects. First, the systems are simultaneously affected by additive noises and multiplicative noises, which makes it impossible to seperately tackle only one form of noises as in \cite{noises_etc_Auto2026} \cite{MS_IF2023} \cite{noises_delay_Auto2019} \cite{noises_Markov_TAC2019}. Second, antagonistic  interactions are detrimental to the formation of consensus in the network. Common results on signed networks with noises focus on bipartite consensus (polarization) \cite{noises_signed_2019} \cite{noises_private_TAC2024} \cite{noises_signed_TSMC2022}, and these works are carried out in the framework of scalar-weighted networks. Third, the inter-dimensional communications amongst multi-dimensional agents, which, depicted by matrix-weight edges, introduce significant difficulty to analyse network dynamics. While \cite{Privacy Preservation_noises_matrix-weight_TCNS2025}	\cite{noises_mw_TCNS2020} studied privacy preservation and consensus performance over matrix-weighted networks with noises, they did not take antagonistic interactions and compound noises as in \eqref{FAN_noise} into consideration.
}
\end{remark}

To achieve non-trivial consensus in \eqref{FAN_noise}, the external input will be exerted to some specified agents in the original matrix-weighted FAN system \eqref{FAN_noise}:
\begin{equation}
\begin{aligned}\label{SAN_noise}
\dot{x}_{i}(t)=&c(t)\sum_{j\in\mathcal{N}_{i}(t)}\left|A_{i j}(t)\right|[\mathbf{sgn}\left(A_{ij}(t)\right) x_{j}(t)-x_{i}(t)+\boldsymbol{1}_{d}\sigma_{ji}\xi_{1ji}(t)+f_{ji}\left(x_{j}(t)-x_{i}(t)\right)\xi_{2ji}(t)]\\
&+c(t)\delta_{i}(t)|B_{i}(t)|[\mathbf{sgn}(B_{i}(t))x_{0}(t)-x_{i}(t)],\ i\in\underline{N},
\end{aligned}
\end{equation}
in which $\delta_{i}(t)\geq0, i\in\underline{N}$ are the coupling coefficients between external signal $x_{0}(t)$ and agent $x_{i}(t)$, and real symmetric matrices $B_{i}(t)\in\mathbb{R}^{d\times d},\ i\in\underline{N}$ are the corresponding coupling
matrix weights. Dynamics \eqref{SAN_noise} establishes a semi-autonomous network (SAN) \cite{FAN&SAN}, in which a subset of agents (referred to as informed agents \cite{naive}) are selected to receive external signals to steer the entire network toward the desired non-zero consensus state. At time $t$, the vertex set of informed agents will be denoted as $\mathcal{V}_{\mathcal{I}}(t)$, for each agent $i, v_{i}\in \mathcal{V}_{\mathcal{I}}(t)$ if and only if $\delta_{i}(t)>0$ and $B_{i}(t)\neq\boldsymbol{0}_{d\times d}$. The naive vertex \cite{naive} set is $\mathcal{V}_{\mathcal{N}}(t)=\mathcal{V}\backslash\mathcal{V}_{\mathcal{I}}(t)$. It is time now to give the specific definitions of non-trivial consensus in both the mean square and the almost sure senses:

\begin{definition}[Mean square non-trivial consensus]
{\rm
Under signed matrix-weighted SAN dynamics \eqref{SAN_noise}, if there exist $\delta_{i}(t)\in\mathbb{R},\ B_{i}(t)\in\mathbb{R}^{d\times d},\ i\in\underline{N}$ and $x_{0}\in\mathbb{R}^{d}$, such that for any initial values $x_{i}(0),\ i\in\underline{N}$,
\begin{align}\label{MS definition}
\lim\limits_{t\rightarrow+\infty}\mathbb{E}\left\Vert x_{i}(t)-\boldsymbol{\theta}\right\Vert^{2}=0,\ i\in\underline{N},
\end{align}
where $\boldsymbol{\theta}\neq\boldsymbol{0}_{d}$ is the preset consensus state, then we say that the mean square non-trivial consensus for system \eqref{SAN_noise} is realized.
}
\end{definition}

\begin{definition}[Almost sure non-trivial consensus]
{\rm
Under signed matrix-weighted SAN dynamics \eqref{SAN_noise}, if there exist $\delta_{i}(t)\in\mathbb{R},\ B_{i}(t)\in\mathbb{R}^{d\times d},\ i\in\underline{N}$ and $x_{0}\in\mathbb{R}^{d}$, such that for any initial values $x_{i}(0),\ i\in\underline{N}$,
\begin{align}\label{AS definition}
\lim\limits_{t\rightarrow+\infty}x_{i}(t)-\boldsymbol{\theta}=0,\ a.s., \ i\in\underline{N},
\end{align}
where $\boldsymbol{\theta}\neq\boldsymbol{0}_{d}$ is the preset consensus state, then we say that the almost sure non-trivial consensus for system \eqref{SAN_noise} is realized.
} 
\end{definition}

Given a desired non-trivial consensus state $\boldsymbol{\theta}\neq\textbf{0}_{d}$ for agents $x_{i},\ i\in\underline{N}$, the main task of this paper is to realize mean square and almost sure non-trivial consensus by appropriately designing the coupling matrix weights $B_{i}(t)$, the couplying coefficients $\delta_{i}(t)$ and the external signal $x_{0}(t)$, including the basic problem of selecting the informed agents set $\mathcal{V}_{\mathcal{I}}(t)\subseteq\mathcal{V}$ to receive external signal $x_{0}(t)$.

We introduce the grounded Laplacian $L_{B}(t)$ as
\begin{align}\label{LB}
L_{B}(t)=L(t)+(\Delta(t)\otimes I_{d})\ldotp\mathbf{diag}(|B(t)|),
\end{align}
where $\Delta(t)=\mathbf{diag}(\delta_{1}(t),...,\delta_{N}(t))$, and $B(t)=\left[B_{1}(t);...;B_{N}(t)\right]\in\mathbb{R}^{Nd\times d}$, and $\mathbf{diag}(|B(t)|)=\mathbf{diag}\left(|B_{1}(t)|,...,|B_{N}(t)|\right)$ is the block diagonal matrix with $|B_{i}(t)|$ being its $i$th diagonal block. Based on \eqref{LB}, signed matrix-weighted SAN dynamics \eqref{SAN_noise} can be written in its compact form:
\begin{equation}\label{compact dynamics}
\begin{aligned}
\dot{x}(t)=&-c(t)L_{B}(t)x(t)+c(t)\left(\Delta(t)\otimes I_{d}\right)B(t)x_{0}(t)\\
&+c(t)\left[\sum_{i,j=1}^{N}\eta_{N,i}\otimes\left(\sigma_{ji}|A_{ij}\left(t\right)|\boldsymbol{1}_{d}\right)\right]\xi_{1ji}(t)\\
&+c(t)\left[\sum_{i,j=1}^{N}\eta_{N,i}\otimes\left[|A_{ij}(t)|f_{ji}\left(x_{j}\left(t\right)-x_{i}\left(t\right)\right)\right]\right]\xi_{2ji}(t)
\end{aligned}
\end{equation}
in which $\eta_{N,i}=[0,...,0,1,0,...,0]^{\top}\in\mathbb{R}^{N}$ with only one nonzero element $1$ in its $i$th position.

Recall that the vertex set consisting of every vertex that has incoming negative edge(s) in $\mathcal{G}(t)$ is denoted as $\mathcal{U}(t)$, $\varOmega_{i}(t)$ is the vertex set consisting of all the neighbor vertices of $v_{i}$ that has outgoing negative edge pointing to $v_{i}$ in $\mathcal{G}(t)$, and $v_{i}$'s in-neighbor vertices set at time $t$ is defined as $\mathcal{N}_{i}(t)=\{v_{j}\in\mathcal{V}|\mathbf{sgn}(A_{ij}(t))\neq0\}$, and the out-neighbor vertices set is $\mathcal{N}'_{i}(t)=\{v_{j}\in\mathcal{V}|\mathbf{sgn}(A_{ji}(t))\neq0\}$. Utilizing the above notations, Lemma \ref{Basic Lemma}, a basic result on the properties of grounded Laplacian $L_{B}(t)$ will be provided, which commences with Definition \ref{path definition} and Definition \ref{in-degree-dominated}.

\begin{definition}[positive-negative path]\label{path definition}
{\rm A path in matrix-weighted network is called a positive-negative path, if every edge in this path is either positive definite or negative definite.}
\end{definition}

\begin{definition}[in-degree-dominated]\label{in-degree-dominated}
{\rm In a signed matrix-weighted graph $\mathcal{G}=(\mathcal{V},\mathcal{E},\mathcal{A})$, vertex $v_{i}\in\mathcal{V}$ is called in-degree-dominated, if there holds}
\begin{align}\label{in-degree-dominated equation}
\sum_{j\in\mathcal{N}_{i}}|A_{ij}|-\sum_{j\in\mathcal{N}'_{i}}|A_{ji}|\succeq\boldsymbol{0}.
\end{align}
{\rm in which $\mathcal{N}_{i}$ and $\mathcal{N}'_{i}$ are the in-neighbor and out-neighbor vertex sets of vertex $v_{i}$, respectively.}
\end{definition}

Definition \ref{in-degree-dominated} includes the special case of $\sum_{j\in\mathcal{N}_{i}}|A_{ij}|=\sum_{j\in\mathcal{N}'_{i}}|A_{ji}|$, which can be regarded as the matrix-weighted version of the concept ``balanced node” proposed in \cite{consensus problems}. For detailed explanations, please refer to Remark 2 in \cite{my_matrix-weighted_NTC}. 

\begin{assumption}\label{directed Assumption}
{\rm For any $t\geq0$ and the directed signed matrix-weighted graph $\mathcal{G}(t)=\left(\mathcal{V},\mathcal{E}(t),\mathcal{A}(t)\right)$ of the original FAN \eqref{SAN_noise}, its vertex set $\mathcal{V}$ can be decomposed into $\mathcal{V}=\mathcal{V}_{1}(t)\cup\mathcal{V}_{2}(t)$ with $\mathcal{V}_{1}(t)\cap\mathcal{V}_{2}(t)=\emptyset$, such that
\begin{enumerate}
\item For any vertex $v_{j}\in\mathcal{V}_{2}(t)$, there exists at least one vertex $v_{i}\in\mathcal{V}_{1}(t)$ and one positive-negative path $\mathcal{P}_{ij}(t)$ from $v_{i}$ to $v_{j}$.
\item Every vertex $v_{j}\in\mathcal{V}_{2}(t)$ is in-degree-dominated.
\end{enumerate}}
\end{assumption}

\begin{lemma}[\cite{my_matrix-weighted_NTC}]\label{Basic Lemma}
{\rm Suppose that Assumption \ref{directed Assumption} holds, and $\sum_{j\in\varOmega_{i}(t)}|A_{ij}(t)|\succ\boldsymbol{0}$. For any $t\geq0$, denote
\begin{equation}\label{C in switching NTC thm}\notag
\begin{aligned}
&C(t)=\max\limits_{i\in\mathcal{V}_{1}(t)}\{C_{i}(t)\},\\
&C_{i}(t)=\frac{1}{2}\lambda_{max}\left[\left(\sum_{j\in\varOmega_{i}(t)}|A_{ij}(t)|\right)^{-1}\left(\sum_{j\in\mathcal{N}'_{i}(t)}|A_{ji}(t)|-\sum_{j\in\mathcal{N}_{i}(t)}|A_{ij}(t)|\right)\right],\\ &i\in\mathcal{V}_{1}(t),
\end{aligned}
\end{equation}
choose the informed agents set as $\mathcal{V}_{\mathcal{I}}(t)=\mathcal{U}(t)$, take $\delta_{i}(t)=\delta(t),\ \forall i\in\mathcal{V}_{\mathcal{I}}(t)$. If the non-zero coupling coefficient $\delta(t)>C(t)$, and the coupling matrix weights $B_{i}(t)$ and external control signal $x_{0}(t)$ is designed as
\begin{align}\notag
x_{0}(t)=k_{1}(t)\boldsymbol{\theta},\ |B_{i}(t)|=\sum_{j\in\varOmega_{i}(t)}|A_{ij}(t)|,\ i\in\mathcal{V}_{\mathcal{I}}(t),
\end{align}
in which $\boldsymbol{\theta}\neq\boldsymbol{0}_{d}$ is the desired consensus state, and $k_{1}(t)=1+\frac{2}{\delta(t)}$. Then $-L_{B}(t)$ is Hurwitz matrix for any $t\geq0$, and 
\begin{align}\label{LBt property}
-L_{B}(t)(\boldsymbol{1}_{N}\otimes\boldsymbol{\theta})+k_{1}(t)\left[\Delta(t)\otimes I_{d}\right]B(t)\boldsymbol{\theta}=0,\ \forall t\geq0.
\end{align}	$\hfill\square$
}
\end{lemma}

\begin{remark}
{\rm
Throughout the remainder of this paper, the setup in Lemma \ref{Basic Lemma} is adopted by default unless stated otherwise. In particular, we shall take $\delta(t)=C(t)+\Delta C$ without loss of generality, in which $\Delta C$ is any positive constant. Lemma \ref{Basic Lemma} offered the specific strategy to select the informed agents set $\mathcal{V}_{\mathcal{I}}(t)$, to design the couplying matrix weights $B_{i}(t)$, the couplying coefficients $\delta_{i}(t)$ and the external signal $x_{0}(t)$ in response to the time-varying topology $\mathcal{G}(t)$. Detailed discussion of Lemma \ref{Basic Lemma} can be found in Theorem 1 and Theorem 2 from \cite{my_matrix-weighted_NTC}. To understand Assumption \ref{directed Assumption} and Lemma \ref{Basic Lemma} more intuitively, please refer to the specific topology and setup in Subsection \ref{simulation_fixed subsection}.
}
\end{remark}

Following Lemma \ref{Basic Lemma}, for the desired consensus state $\boldsymbol{\theta}\neq\boldsymbol{0}_{d}$, let $\varepsilon(t)\triangleq x(t)-\left(\boldsymbol{1}_{N}\otimes\boldsymbol{\theta}\right)$, by \eqref{LBt property} one has \eqref{compact dynamics} in the form
\begin{equation}\label{epsilon}
\begin{aligned}
d\varepsilon(t)=&-c(t)L_{B}(t)\left[\varepsilon(t)+(\boldsymbol{1}_{N}\otimes\boldsymbol{\theta})\right]dt+c(t)k_{1}(t)\left[\Delta(t)\otimes I_{d}\right]B(t)\boldsymbol{\theta}dt\\
&+c(t)\left[\sum_{i,j=1}^{N}\eta_{N,i}\otimes\left(\sigma_{ji}|A_{ij}\left(t\right)|\boldsymbol{1}_{d}\right)\right]dw_{1ji}(t)\\
&+c(t)\left[\sum_{i,j=1}^{N}\eta_{N,i}\otimes\left[|A_{ij}(t)|f_{ji}\left(x_{j}\left(t\right)-x_{i}\left(t\right)\right)\right]\right]dw_{2ji}(t)\\
=&-c(t)L_{B}(t)\varepsilon(t)dt+c(t)\left[\sum_{i,j=1}^{N}\eta_{N,i}\otimes\left(\sigma_{ji}|A_{ij}\left(t\right)|\boldsymbol{1}_{d}\right)\right]dw_{1ji}(t)\\
&+c(t)\left[\sum_{i,j=1}^{N}\eta_{N,i}\otimes\left[|A_{ij}(t)|f_{ji}\left(x_{j}\left(t\right)-x_{i}\left(t\right)\right)\right]\right]dw_{2ji}(t).
\end{aligned}
\end{equation}
In the subsequent analysis, $\varepsilon(t)$ will be refered to as the non-trivial consensus error, which characterize the difference between agents states and the desired non-trivial consensus value.

\section{Main results}\label{Main results}
\subsection{Fixed topology}
We begin with the fixed topology case, i.e. $\mathcal{G}(t)\equiv\mathcal{G}$, and then extend it to the time-varying case.
\begin{theorem}\label{ms theorem}
{\rm
Consider system \eqref{SAN_noise} with fixed topology $\mathcal{G}$, suppose that Assumption \ref{Lipschitz Assumption} holds for the multiplicative noises intensity $f_{ji}(\cdot)$, and Assumption \ref{directed Assumption} holds for topology $\mathcal{G}$. if the control gain $c(t)$ satisfies
\begin{align}\label{ct in th1}
\int_{0}^{\infty}c(t)dt=\infty,\ \lim\limits_{t\rightarrow+\infty}c(t)=0,
\end{align}
then the mean square non-trivial consensus can be achieved, i.e., \eqref{MS definition} holds for system \eqref{SAN_noise}.
}
\end{theorem}

{\it Proof:}
By \eqref{epsilon}, one can equivalently prove that $\lim\limits_{t\rightarrow+\infty}\mathbb{E}\Vert\varepsilon(t)\Vert^{2}=0$. By Lemma \ref{Basic Lemma}, $-L_{B}$ is Hurwitz, and thus there exists postive definite matrix $P$ such that
\begin{align}\label{LBP}
-PL_{B}-L_{B}^{\top}P=-I_{Nn}
\end{align}
Take Lyapunov function $V(t)=\varepsilon(t)^{\top}P\varepsilon(t)$, then by It$\hat{o}$ formula, one has
\begin{equation}\label{dVt in th1}
\begin{aligned}
dV(t)=&-2c(t)\varepsilon^{\top}(t)PL_{B}\varepsilon(t)dt+dM_{1}(t)\\
&+c^{2}(t)\sum_{i,j=1}^{N}\left(\eta_{N,i}\otimes\sigma_{ji}|A_{ij}|\boldsymbol{1}_{d}\right)^{\top}P\left(\eta_{N,i}\otimes\sigma_{ji}|A_{ij}|\boldsymbol{1}_{d}\right)dt\\
&+c^{2}(t)\sum_{i,j=1}^{N}\left[\eta_{N,i}\otimes|A_{ij}|f_{ji}\left(x_{j}(t)-x_{i}(t)\right)\right]^{\top}P\left[\eta_{N,i}\otimes|A_{ij}|f_{ji}\left(x_{j}(t)-x_{i}(t)\right)\right]dt,
\end{aligned}
\end{equation}
in which
\begin{equation}\notag
\begin{aligned}
dM_{1}(t)=&2c(t)\varepsilon^{\top}(t)P\sum_{i,j=1}^{N}\left(\eta_{N,i}\otimes\sigma_{ji}|A_{ij}|\boldsymbol{1}_{d}\right)dw_{1ji}(t)\\
&+2c(t)\varepsilon^{\top}(t)P\sum_{i,j=1}^{N}\left[\eta_{N,i}\otimes|A_{ij}|f_{ji}\left(x_{j}(t)-x_{i}(t)\right)\right]dw_{2ji}(t).
\end{aligned}
\end{equation} 
Then
\begin{equation}\notag
\begin{aligned}
dV(t)\leq&-c(t)\varepsilon(t)^{\top}\varepsilon(t)dt+dM_{1}(t)+K_{N,P}\cdot c^{2}(t)dt\\
&+\bar{\sigma}Q_{N,P}\cdot c^{2}(t)\left[2\varepsilon^{\top}(t)\left(S_{N}\otimes I_{n}\right)\varepsilon(t)\right],
\end{aligned}
\end{equation}
in which $K_{N,P},\ Q_{N,P}$ are positive constants, $S_{N}\triangleq (NI_{N}-J_{N})\otimes I_{d}$, $J_{N}=\boldsymbol{1}_{N}\boldsymbol{1}_{N}^{\top}$. Then
\begin{equation}\notag
\begin{aligned}
dV(t)\leq&-c(t)\varepsilon^{\top}(t)P^{\frac{1}{2}}(P^{-1})P^{\frac{1}{2}}\varepsilon(t)dt+dM_{1}(t)+K_{N,P}\cdot c^{2}(t)dt\\
&+2\bar{\sigma}Q_{N,P}\cdot c^{2}(t)\varepsilon^{\top}(t)P^{\frac{1}{2}}\left(P^{-\frac{1}{2}}S_{N}P^{-\frac{1}{2}}\right)P^{\frac{1}{2}}\varepsilon(t),
\end{aligned}
\end{equation}
which gives that
\begin{align}\notag
dV(t)\leq-\mu_{1}(t)V(t)dt+dM(t)+K_{N,P}c^{2}(t)dt,
\end{align}
where $\mu_{1}(t)\triangleq c(t)\left[\lambda_{min}(P^{-1})-2\bar{\sigma}Q_{N,P}\lambda_{max}(P^{-\frac{1}{2}}S_{N}P^{-\frac{1}{2}})c(t)\right]>0,\ \forall t\geq t_{0}$ for some $t_{0}\geq0$, since the control gain function $c(t)$ satisfies $\lim\limits_{t\rightarrow+\infty}c(t)=0$. Let $U(t)=e^{\int_{t_{0}}^{t}\mu_{1}(s)ds}V(t)$, then
\begin{equation}\notag
\begin{aligned}
dU(t)&=\mu_{1}(t)e^{\int_{t_{0}}^{t}\mu_{1}(s)ds}V(t)dt+e^{\int_{t_{0}}^{t}\mu_{1}(s)ds}dV(t)\\
&\leq e^{\int_{t_{0}}^{t}\mu_{1}(s)ds}\left[dM(t)+K_{N,P}c^{2}(t)dt\right],
\end{aligned}
\end{equation}
which leads to
\begin{align}\notag
U(t)\leq U(t_{0})+K_{N,P}\int_{t_{0}}^{t}e^{\int_{t_{0}}^{s}\mu_{1}(u)du}c^{2}(s)ds+\int_{t_{0}}^{t}e^{\int_{t_{0}}^{s}\mu_{1}(u)du}dM_{1}(s).
\end{align}
Take expectations on both sides:
\begin{align}\notag
\mathbb{E}(U(t))\leq\mathbb{E}(U(t_{0}))+K_{N,P}\int_{t_{0}}^{t}e^{\int_{t_{0}}^{s}\mu_{1}(u)du}c^{2}(s)ds,
\end{align}
then
\begin{align}\notag
\mathbb{E}(V(t))\leq e^{-\int_{t_{0}}^{t}\mu_{1}(s)ds}\mathbb{E}(V(t_{0}))+K_{N,P}\int_{t_{0}}^{t}e^{-\int_{s}^{t}\mu_{1}(u)du}c^{2}(s)ds,
\end{align}
in which by L'H$\hat{o}$pital's Rule, one has
\begin{equation}\notag
\begin{aligned}
\lim\limits_{t\rightarrow+\infty}\int_{t_{0}}^{t}e^{-\int_{s}^{t}\mu_{1}(u)du}c^{2}(s)ds&=\lim\limits_{t\rightarrow+\infty}\dfrac{\int_{t_{0}}^{t}e^{\int_{0}^{s}\mu_{1}(u)du}c^{2}(s)ds}{e^{\int_{0}^{t}\mu_{1}(u)du}}\\
&=\lim\limits_{t\rightarrow+\infty}\dfrac{e^{\int_{0}^{t}\mu_{1}(u)du}c^{2}(t)}{\mu_{1}(t)\cdot e^{\int_{0}^{t}\mu_{1}(u)du}}\\
&=\lim\limits_{t\rightarrow+\infty}\dfrac{c^{2}(t)}{\mu_{1}(t)}=0,
\end{aligned}
\end{equation}
in which $\int_{0}^{\infty}\mu_{1}(u)du=\infty$ since $\int_{0}^{\infty}c(t)dt=\infty$. Therefore, $\lim\limits_{t\rightarrow+\infty}\mathbb{E}(V(t))=0$, which combining with the positive definiteness of $P$ leads to $\lim\limits_{t\rightarrow+\infty}\mathbb{E}\Vert\varepsilon(t)\Vert^{2}=0$.
$\hfill\blacksquare$

In the following, we examine the almost sure non-trivial consensus result. For this purpose, the semimartingale convergence theorem is introduced.

\begin{lemma}[\cite{semimartingale}, semimartingale convergence theorem]\label{sc lemma}
{\rm
Let $Y_{1}(t)$ and $Y_{2}(t)$ be the two adapted increasing processes on $t\geq0$ with $Y_{1}(0)=Y_{2}(0)=0,\ a.s.$ Let $M(t)$ be a real-valued local martingale with $M(0)=0,\ a.s.$, and let $\zeta$ be a nonnegative measurable random variable. Assume that $X(t)$ is nonnegative and
\begin{align}\notag
X(t)=\zeta+Y_{1}(t)-Y_{2}(t)+M(t),\ t\geq0.
\end{align}
If $\lim\limits_{t\rightarrow+\infty}Y_{1}(t)<\infty,a.s.$, then for almost all $\omega\in\Omega$,
$$\lim\limits_{t\rightarrow+\infty}X(t)<\infty\ \text{and}\ \lim\limits_{t\rightarrow+\infty}Y_{2}(t)<\infty.$$
}
\end{lemma}

\begin{theorem}\label{AS theorem}
{\rm
Consider system \eqref{SAN_noise} with fixed topology $\mathcal{G}$, suppose that Assumption \ref{Lipschitz Assumption} holds for the multiplicative noises intensity $f_{ji}(\cdot)$ and Assumption \ref{directed Assumption} holds for topology $\mathcal{G}$. if the control gain $c(t)$ satisfies
\begin{align}\label{ct in th2}
\int_{0}^{\infty}c(t)dt=\infty,\ \lim\limits_{t\rightarrow+\infty}c(t)=0,\ \int_{0}^{\infty}c^{2}(t)dt<\infty,
\end{align}
then the almost sure non-trivial consensus can be achieved, i.e., \eqref{AS definition} holds for system \eqref{SAN_noise}.
}
\end{theorem}

{\it Proof:}
By \eqref{dVt in th1},
\begin{equation}\notag
\begin{aligned}
V(t)=&V(t_{0})+\sum_{i,j=1}^{N}\left(\eta_{N,i}\otimes\sigma_{ji}|A_{ij}|\boldsymbol{1}_{d}\right)^{\top}P\left(\eta_{N,i}\otimes\sigma_{ji}|A_{ij}|\boldsymbol{1}_{d}\right)\int_{t_0}^{t}c^{2}(s)ds\\
&+\sum_{i,j=1}^{N}\int_{t_0}^{t}\left[\eta_{N,i}\otimes|A_{ij}|f_{ji}\left(x_{j}(s)-x_{i}(s)\right)\right]^{\top}P\left[\eta_{N,i}\otimes|A_{ij}|f_{ji}\left(x_{j}(s)-x_{i}(s)\right)\right]\cdot c^{2}(s)ds\\
&-2\int_{t_0}^{t}c(s)\varepsilon^{\top}(s)PL_{B}\varepsilon(s)ds\\
&+2\sum_{i,j=1}^{N}\int_{t_0}^{t}c(s)\varepsilon^{\top}(s)P\left(\eta_{N,i}\otimes\sigma_{ji}|A_{ij}|\boldsymbol{1}_{d}\right)dw_{1ji}(s)\\
&+2\sum_{i,j=1}^{N}\int_{t_0}^{t}c(s)\varepsilon^{\top}(s)P\left[\eta_{N,i}\otimes|A_{ij}|f_{ji}\left(x_{j}(s)-x_{i}(s)\right)\right]dw_{2ji}(s)
\end{aligned}
\end{equation}
Following notations in Lemma \ref{sc lemma}, one can denote
\begin{equation}\notag
\begin{aligned}
Y_{1}(t)=&\sum_{i,j=1}^{N}\left(\eta_{N,i}\otimes\sigma_{ji}|A_{ij}|\boldsymbol{1}_{d}\right)^{\top}P\left(\eta_{N,i}\otimes\sigma_{ji}|A_{ij}|\boldsymbol{1}_{d}\right)\int_{t_0}^{t}c^{2}(s)ds,\\
Y_{2}(t)=&-\sum_{i,j=1}^{N}\int_{t_0}^{t}\left[\eta_{N,i}\otimes|A_{ij}|f_{ji}\left(x_{j}(s)-x_{i}(s)\right)\right]^{\top}P\left[\eta_{N,i}\otimes|A_{ij}|f_{ji}\left(x_{j}(s)-x_{i}(s)\right)\right]\cdot c^{2}(s)ds\\
&+2\int_{t_0}^{t}c(s)\varepsilon^{\top}(s)PL_{B}\varepsilon(s)ds,\\
M(t)=&2\sum_{i,j=1}^{N}\int_{t_0}^{t}c(s)\varepsilon^{\top}(s)P\left(\eta_{N,i}\otimes\sigma_{ji}|A_{ij}|\boldsymbol{1}_{d}\right)dw_{1ji}(s)\\
&+2\sum_{i,j=1}^{N}\int_{t_0}^{t}c(s)\varepsilon^{\top}(s)P\left[\eta_{N,i}\otimes|A_{ij}|f_{ji}\left(x_{j}(s)-x_{i}(s)\right)\right]dw_{2ji}(s),
\end{aligned}
\end{equation}
and $\zeta=V(t_0)$. In which $Y_{1}(t)$ is increasing process on $t\geq t_{0}$ with $Y_{1}(t_{0})=0,\ a.s.$ and $\lim\limits_{t\rightarrow+\infty}Y_{1}(t)<\infty,a.s.$ since $\int_{0}^{\infty}c^{2}(t)dt<\infty$.
\begin{equation}\notag
\begin{aligned}
Y_{2}(t)=&\int_{t_0}^{t}c(s)\Vert\varepsilon(t)\Vert^2\\
&-\sum_{i,j=1}^{N}\left[\eta_{N,i}\otimes|A_{ij}|f_{ji}\left(x_{j}(s)-x_{i}(s)\right)\right]^{\top}P\left[\eta_{N,i}\otimes|A_{ij}|f_{ji}\left(x_{j}(s)-x_{i}(s)\right)\right]\cdot c^{2}(s)ds\\
\geq&\int_{t_0}^{t}c(s)\left[\lambda_{min}(P^{-1})-2\bar{\sigma}Q_{N,P}\lambda_{max}(P^{-\frac{1}{2}S_{N}P^{-\frac{1}{2}}})\right]V(s)ds\\
=&\int_{t_0}^{t}\mu(s)V(s)ds,
\end{aligned}
\end{equation}
since $\mu(s)>0,\ \forall s\geq t_{0}$, one has that $Y_{2}(t)$ is increasing process with $Y_{2}(t_0)=0,\ a.s.$\\
Take stopping times 
\begin{align}\notag
\tau_{n}=inf\{\Vert\varepsilon(t)\Vert\geq n\ \text{or}\ \int_{0}^{t}c^{2}(s)ds\geq n\},	
\end{align}
then by continuity one has that $\lim\limits_{t\rightarrow+\infty}\tau_{n}=+\infty$. Obviously,
\begin{align}\notag
\mathbb{E}\int_{0}^{t\wedge\tau_{n}}H_{1ji}^{2}(s)ds<\infty,\ \mathbb{E}\int_{0}^{t\wedge\tau_{n}}H_{2ji}^{2}(s)ds<\infty,
\end{align}
in which
\begin{align}\notag
H_{1ji}(s)=c(s)\varepsilon^{\top}(s)P\left(\eta_{N,i}\otimes\sigma_{ji}|A_{ij}|\boldsymbol{1}_{d}\right),
\end{align}
and
\begin{align}\notag
H_{2ji}(s)&=c(s)\varepsilon^{\top}(s)P\left[\eta_{N,i}\otimes|A_{ij}|f_{ji}\left(x_{j}(s)-x_{i}(s)\right)\right].
\end{align}

Therefore, $M(t)$ is real-valued local martingale with $M(0)=0,\ a.s.$ By Lemma \ref{sc lemma} one has that $\lim\limits_{t\rightarrow+\infty}V(t)<\infty,\ a.s.$ Theorem \ref{ms theorem} shows that $\lim\limits_{t\rightarrow+\infty}\mathbb{E}(V(t))=0$, which implies that there exists a subsequence converging to $0$ almost surely. The uniqueness of limit leads to $\lim\limits_{t\rightarrow+\infty}V(t)=0,\ a.s.$ Equivalently, one has that
\begin{align}\notag
\lim\limits_{t\rightarrow+\infty}x_{i}(t)-\boldsymbol{\theta}=0,\ a.s.,\ i\in\underline{N}.
\end{align}
$\hfill\blacksquare$

\begin{remark}\label{ct Remark}
{\rm
The design of control gain $c(t)$ is essential for realizing the desired non-trivial consensus result. In conditions \eqref{ct in th2}, first, $\int_{0}^{\infty}c(t)dt=\infty$ ensures the cumulative effect of the consensus protocol does not vanish over time, which allows the agents to exchange sufficient information so that their state errors can be persistently reduced. Second, $\lim\limits_{t\rightarrow+\infty}c(t)=0$ requires the interaction gain to decay asymptotically, which gradually attenuating the influence of measurement noises. As time evolves, a diminishing gain prevents persistent random fluctuations from dominating the system dynamics and destabilizing the asymptotic behavior. Last, condition $\int_{0}^{\infty}c^{2}(t)dt<\infty$ characterizes the finite energy of the noise-induced effects. Since the stochastic disturbances enter the system multiplied by $c(t)$, this condition ensures that the accumulated variance contributed by the noises remains bounded over time. These conditions play a crucial role in guaranteeing mean square and almost sure non-trivial consensus for the signed matrix-weighted networks under compound noises in our study.
}
\end{remark}

\subsection{Time-varying topology}

For the time-varying topology $\mathcal{G}(t)=(\mathcal{V},\mathcal{E}(t),\mathcal{A}(t))$, the following Assumption on the boundedness of elements in edge weight matrices is proposed.

\begin{assumption}\label{bound assumption}
{\rm There exists a positive constant $D>0$ such that
\begin{align}\notag
|a_{kl}^{ij}(t)|\leq D,\ \forall i,j\in\underline{N},\ k,l\in\underline{d},\ t\geq0,
\end{align}
where $a_{kl}^{ij}(t)$ is the $(k,l)$th element of matrix $A_{ij}(t)$.
}
\end{assumption}

\begin{remark}
{\rm
Assumption \ref{bound assumption} serves as a relatively mild condition for networks with time-varying topologies. In relevant achievements \cite{noises_SIAM2018} \cite{switching matrix-weighted consensus} \cite{switching matrix-weighted cluster consensus}, topologies ``$\mathcal{G}(t)$"s are chosen from a finite set. Besides, in \cite{switching matrix-weighted consensus} \cite{switching matrix-weighted cluster consensus} that studied consensus and cluster consensus on switching matrix-weighted networks,   ``$\mathcal{G}(t)$''s need to hold constants in every dwell interval ``$\left[t_{k},t_{k+1}\right)$'', in which the dwell time ``$\Delta t_{k}=t_{k+1}-t_{k}$'' are lower bounded by a positive constant. By comparison, Assumption \ref{bound assumption} promotes the applicability of concept ``time-varying topology" in the matrix weighted network consensus algorithms field.
}
\end{remark}

\begin{theorem}\label{switching ms theorem}
{\rm
Consider system \eqref{SAN_noise} with time-varying topology $\mathcal{G}(t)$, suppose that Assumption \ref{Lipschitz Assumption} holds for the multiplicative noises intensity functions $f_{ji}(\cdot)$ and Assumption \ref{directed Assumption}, Asuumption \ref{bound assumption} hold for topology $\mathcal{G}(t)$. if the control gain $c(t)$ satisfies
\begin{align}\label{ct in th3}
\int_{0}^{\infty}c(t)dt=\infty,\ \lim\limits_{t\rightarrow+\infty}c(t)=0,
\end{align}
and matrix $\left(L_{B}(t)+L_{B}^{\top}(t)\right)$ is positive definite fot any $t\geq0$, then the mean square non-trivial consensus can be achieved, i.e., \eqref{MS definition} holds for system \eqref{SAN_noise}.
}
\end{theorem}

{\it Proof:}
Take Lyapunov function $V(t)=\varepsilon(t)^{\top}\varepsilon(t)$, then by It$\hat{o}$ formula, one has
\begin{equation}\label{dVt in th3}
\begin{aligned}
dV(t)=&-2c(t)\varepsilon^{\top}(t)L_{B}(t)\varepsilon(t)dt+dM_{2}(t)\\
&+c^{2}(t)\sum_{i,j=1}^{N}\left(\eta_{N,i}\otimes\sigma_{ji}|A_{ij}(t)|\boldsymbol{1}_{d}\right)^{\top}\left(\eta_{N,i}\otimes\sigma_{ji}|A_{ij}(t)|\boldsymbol{1}_{d}\right)dt\\
&+c^{2}(t)\sum_{i,j=1}^{N}\left[\eta_{N,i}\otimes|A_{ij}(t)|f_{ji}\left(x_{j}(t)-x_{i}(t)\right)\right]^{\top}\left[\eta_{N,i}\otimes|A_{ij}(t)|f_{ji}\left(x_{j}(t)-x_{i}(t)\right)\right]dt,
\end{aligned}
\end{equation}

in which $dM_{2}(t)=2c(t)\varepsilon^{\top}(t)\sum_{i,j=1}^{N}\left(\eta_{N,i}\otimes\sigma_{ji}|A_{ij}(t)|\boldsymbol{1}_{d}\right)dw_{1ji}(t)+2c(t)\varepsilon^{\top}(t)\\
\sum_{i,j=1}^{N}\left[\eta_{N,i}\otimes|A_{ij}(t)|f_{ji}\left(x_{j}(t)-x_{i}(t)\right)\right]dw_{2ji}(t)$. Therefore,
\begin{equation}\label{varepsilon in th3}
\begin{aligned}
d\Vert\varepsilon(t)\Vert^{2}=&-2c(t)\varepsilon^{\top}(t)L_{B}(t)\varepsilon(t)dt+dM_{2}(t)+c^{2}(t)\cdot\sum_{i,j=1}^{N}\sigma_{ji}^{2}\left[\sum_{l=1}^{d}\left(\sum_{k=1}^{d}a_{lk}^{|ij|}(t)\right)^{2}\right]dt\\
&+c^{2}(t)\sum_{i,j=1}^{N}f_{ji}^{\top}\left(x_{j}(t)-x_{i}(t)\right)|A_{ij}(t)|^{2}f_{ji}\left(x_{j}(t)-x_{i}(t)\right)dt,
\end{aligned}
\end{equation}
in which $a_{lk}^{|ij|}(t)$ is the $(l,k)$th element of matrix $|A_{ij}(t)|$. By Assumption \ref{bound assumption} and the continuous dependence of matrix eigenvalues on the entries, there exist positive constants $\bar{p},\ \Lambda_{1}$ such that
\begin{align}\notag
\sum_{i,j=1}^{N}\sigma_{ji}^{2}\left[\sum_{l=1}^{d}\left(\sum_{k=1}^{d}a_{lk}^{|ij|}(t)\right)^{2}\right]\leq\bar{p},\ \lambda_{max}^{2}(|A_{ij}(t)|)\leq\Lambda_{1},\ \forall i,j=1,...,N,\ \forall t\geq0,
\end{align}
and positive constant $\Lambda_{2}$ such that
\begin{align}\notag
\lambda_{min}\left(L_{B}(t)+L_{B}^{\top}(t)\right)\geq\Lambda_{2}
\end{align}
since $\left(L_{B}(t)+L_{B}^{\top}(t)\right)\succ\boldsymbol{0}$. Therefore, one has
\begin{equation}\notag
\begin{aligned}
d\Vert\varepsilon(t)\Vert^{2}\leq&-c(t)\varepsilon^{\top}(t)\left(L_{B}(t)+L_{B}^{\top}(t)\right)\varepsilon(t)dt+dM_{2}(t)+\bar{p}c^{2}(t)dt\\
&+\bar{\sigma}\Lambda_{1}\cdot c^{2}(t)\cdot\sum_{i,j=1}^{N}\Vert\varepsilon_{i}(t)-\varepsilon_{j}(t)\Vert^{2}dt\\
\leq&-\Lambda_{2}c(t)\Vert\varepsilon(t)\Vert^{2}dt+dM_{2}(t)+\bar{p}c^{2}(t)dt+2\bar{\sigma}\Lambda_{1}\cdot\lambda_{max}(S_{N})\cdot c^{2}(t)\Vert\varepsilon(t)\Vert^{2}dt\\
=&-\mu_{2}(t)\Vert\varepsilon(t)\Vert^{2}dt+\bar{p}c^{2}(t)dt+dM_{2}(t),
\end{aligned}
\end{equation}

in which $S_{N}\triangleq (NI_{N}-J_{N})\otimes I_{d},\ J_{N}=\boldsymbol{1}_{N}\boldsymbol{1}_{N}^{\top}$, and
\begin{align}\notag
\mu_{2}(t)=c(t)\left[\Lambda_{2}-2\bar{\sigma}\Lambda_{1}\lambda_{max}(S_{N})c(t)\right]=c(t)\left[\Lambda_{2}-2\bar{\sigma}\Lambda_{1}Nc(t)\right].
\end{align}
By \eqref{ct in th3}, there exists $t_{0}\geq0$ such that $\mu_{2}(t)>0,\ \forall t\geq t_{0}$. Similar to the proof of Theorem \ref{ms theorem}, one can deduce that $\lim\limits_{t\rightarrow+\infty}\mathbb{E}\Vert\varepsilon(t)\Vert^{2}=0$.
$\hfill\blacksquare$

\begin{remark}
{\rm
According to Theorem \ref{switching ms theorem}, establishing mean square non-trivial consensus for time-varying networks requires the matrix $\left(L_{B}(t)+L_{B}^{\top}(t)\right)$ to be positive definite. In fact, the time-varying topology makes it scarcely possible to derive a common matrix $P\succ\boldsymbol{0}$ with \eqref{LBP} holds for $L_{B}(t),\ t\in\left[0,+\infty\right)$. Therefore, a different form of Lyapunov function is utilized in Theorem \ref{switching ms theorem}. By Lemma \ref{Basic Lemma}, the condition $\left(L_{B}(t)+L_{B}^{\top}(t)\right)\succ\boldsymbol{0}$ can be satisfied provided that $\mathcal{G}(t)$ is undirected. Moreover, this condition holds for all the directed topologies adopted for simulation examples in Subsection \ref{simulation_varying}, and they were not specially selected for this property.
}
\end{remark}

\begin{theorem}\label{switching AS theorem}
{\rm
Consider system \eqref{SAN_noise} with time-varying topology $\mathcal{G}(t)$, suppose that Assumption \ref{Lipschitz Assumption} holds for the multiplicative noises intensity $f_{ji}(\cdot)$ and Assumption \ref{directed Assumption}, Assumption \ref{bound assumption} hold for topology $\mathcal{G}(t)$. if the control gain $c(t)$ satisfies
\begin{align}\label{ct in th4}
\int_{0}^{\infty}c(t)dt=\infty,\ \lim\limits_{t\rightarrow+\infty}c(t)=0,\ \int_{0}^{\infty}c^{2}(t)dt<\infty,
\end{align}
and matrix $\left(L_{B}(t)+L_{B}^{\top}(t)\right)$ is positive definite, then the almost sure non-trivial consensus can be achieved, i.e., \eqref{AS definition} holds for system \eqref{SAN_noise}.
}
\end{theorem}

{\it Proof:}
\eqref{varepsilon in th3} gives
\begin{equation}\notag
\begin{aligned}
\Vert\varepsilon(t)\Vert^{2}=&\Vert\varepsilon(t_{0})\Vert^{2}-\int_{t_0}^{t}c(s)\varepsilon^{\top}(s)\left(L_{B}(s)+L_{B}^{\top}(s)\right)\varepsilon(s)ds+\int_{t_0}^{t}dM_{2}(s)\\
&+\int_{t_0}^{t}c^{2}(s)\cdot\sum_{i,j=1}^{N}\sigma_{ji}^{2}\left[\sum_{l=1}^{d}\left(\sum_{k=1}^{d}a_{lk}^{|ij|}(s)\right)^{2}\right]ds\\
&+\int_{t_0}^{t}c^{2}(s)\sum_{i,j=1}^{N}f_{ji}^{\top}\left(x_{j}(s)-x_{i}(s)\right)|A_{ij}(s)|^{2}f_{ji}\left(x_{j}(s)-x_{i}(s)\right)ds.
\end{aligned}
\end{equation}
Denote
\begin{equation}\notag
\begin{aligned}
Y_{1}(t)=&\int_{t_0}^{t}c^{2}(s)\cdot\sum_{i,j=1}^{N}\sigma_{ji}^{2}\left[\sum_{l=1}^{d}\left(\sum_{k=1}^{d}a_{lk}^{|ij|}(s)\right)^{2}\right]ds,\\
Y_{2}(t)=&-\int_{t_0}^{t}c^{2}(s)\sum_{i,j=1}^{N}f_{ji}^{\top}\left(x_{j}(s)-x_{i}(s)\right)|A_{ij}(s)|^{2}f_{ji}\left(x_{j}(s)-x_{i}(s)\right)ds\\
&+\int_{t_0}^{t}c(s)\varepsilon^{\top}(s)\left(L_{B}(s)+L_{B}^{\top}(s)\right)\varepsilon(s)ds,\\
M(t)=&\int_{t_0}^{t}dM_{2}(s)\\
=&2\int_{t_0}^{t}c(s)\varepsilon^{\top}(s)\sum_{i,j=1}^{N}\left(\eta_{N,i}\otimes\sigma_{ji}|A_{ij}(s)|\boldsymbol{1}_{d}\right)dw_{1ji}(s)\\
&+2\int_{t_0}^{t}c(s)\varepsilon^{\top}(s)\sum_{i,j=1}^{N}\left[\eta_{N,i}\otimes|A_{ij}(s)|f_{ji}\left(x_{j}(s)-x_{i}(s)\right)\right]dw_{2ji}(s)
\end{aligned}
\end{equation}
and
$\zeta=\Vert\varepsilon(t_{0})\Vert^{2}$. It can be verified that $Y_{1}(t)$ is an increasing process with $Y_{1}(t_0)=0,\ a.s.$, and $\lim\limits_{t\rightarrow+\infty}Y_{1}(t)\leq\bar{p}\int_{t_0}^{\infty}c^{2}(s)ds<\infty,\ a.s.$, and
\begin{align}\notag
Y_{2}(t)\geq\int_{t_0}^{t}c(s)\left[\Lambda_{2}-2\bar{\sigma}\Lambda_{1}\lambda_{max}(S_{N})c(s)\right]\Vert\varepsilon(s)\Vert^{2}ds=\int_{t_0}^{t}\mu_{2}(s)\Vert\varepsilon(s)\Vert^{2}ds,
\end{align}
in which $\mu_{2}(t)>0,\ \forall t\geq t_{0}$. Therefore, $Y_{2}(t)$ is an increasing process with $Y_{2}(t_0)=0,\ a.s.$ Denote
\begin{align}\notag
H_{1ji}(s)=c(s)\varepsilon^{\top}(s)\left(\eta_{N,i}\otimes\sigma_{ji}|A_{ij}(s)|\boldsymbol{1}_{d}\right), 
\end{align} 
and
\begin{align}\notag
H_{2ji}(s)=c(s)\varepsilon^{\top}(s)\left[\eta_{N,i}\otimes|A_{ij}(s)|f_{ji}\left(x_{j}(s)-x_{i}(s)\right)\right].
\end{align}
Take stopping times
\begin{align}\label{stopping time in th4}
\tau_{n}=inf\{\Vert\varepsilon(t)\Vert\geq n\ \text{or}\ \int_{0}^{t}c^{2}(s)ds\geq n\},	
\end{align}
then by continuity one has that $\lim\limits_{n\rightarrow+\infty}\tau_{n}=+\infty$. Further, by Assumption \ref{Lipschitz Assumption}, Assumption \ref{bound assumption} and \eqref{stopping time in th4},
\begin{equation}\notag
\begin{aligned}
&\mathbb{E}\int_{0}^{t\wedge\tau_{n}}H^{2}_{1ji}(s)ds\\
=&\mathbb{E}\int_{0}^{t\wedge\tau_{n}}c^{2}(s)\varepsilon^{\top}(s)\left(\eta_{N,i}\otimes\sigma_{ji}|A_{ij}(s)|\boldsymbol{1}_{d}\right)\left(\eta_{N,i}\otimes\sigma_{ji}|A_{ij}(s)|\boldsymbol{1}_{d}\right)^{\top}\varepsilon(s)ds\\
\triangleq&\mathbb{E}\int_{0}^{t\wedge\tau_{n}}c^{2}(s)\varepsilon^{\top}(s)\mathcal{A}_{1ji}(s)\varepsilon(s)ds\\
\leq&\mathbb{E}\int_{0}^{t\wedge\tau_{n}}\lambda_{max}(\mathcal{A}_{1ji}(s))\cdot c^{2}(s)\Vert\varepsilon(s)\Vert^{2}ds<\infty,
\end{aligned}
\end{equation}
and
\begin{equation}\notag
\begin{aligned}
&\mathbb{E}\int_{0}^{t\wedge\tau_{n}}H^{2}_{2ji}(s)ds\\
=&\mathbb{E}\int_{0}^{t\wedge\tau_{n}}c^{2}(s)\varepsilon^{\top}(s)\left(\eta_{N,i}\otimes|A_{ij}(s)|f_{ji}\left(x_{j}(s)-x_{i}(s)\right)\right)\left(\eta_{N,i}\otimes|A_{ij}(s)|f_{ji}\left(x_{j}(s)-x_{i}(s)\right)\right)^{\top}\varepsilon(s)ds\\
=&\mathbb{E}\int_{0}^{t\wedge\tau_{n}}c^{2}(s)\varepsilon_{i}^{\top}(s)|A_{ij}(s)|\cdot f_{ji}\left(x_{j}(s)-x_{i}(s)\right)f_{ji}^{\top}\left(x_{j}(s)-x_{i}(s)\right)|A_{ij}(s)|\varepsilon_{i}(s)ds\\
=&\mathbb{E}\int_{0}^{t\wedge\tau_{n}}c^{2}(s)\varepsilon_{i}^{\top}(s)|A_{ij}(s)|\cdot f_{ji}\left(\varepsilon_{j}(s)-\varepsilon_{i}(s)\right)f_{ji}^{\top}\left(\varepsilon_{j}(s)-\varepsilon_{i}(s)\right)|A_{ij}(s)|\varepsilon_{i}(s)ds\\
\leq&\mathbb{E}\int_{0}^{t\wedge\tau_{n}}\lambda_{max}(\mathcal{A}_{2ji}(s))\cdot c^{2}(s)\Vert\varepsilon_{i}(s)\Vert^{2}ds<\infty.
\end{aligned}
\end{equation}
in which $\mathcal{A}_{1ji}(s)=\left(\eta_{N,i}\otimes\sigma_{ji}|A_{ij}(s)|\boldsymbol{1}_{d}\right)\left(\eta_{N,i}\otimes\sigma_{ji}|A_{ij}(s)|\boldsymbol{1}_{d}\right)^{\top}$ and $\mathcal{A}_{2ji}(s)=|A_{ij}(s)|\cdot f_{ji}\left(\varepsilon_{j}(s)-\varepsilon_{i}(s)\right)f_{ji}^{\top}\left(\varepsilon_{j}(s)-\varepsilon_{i}(s)\right)|A_{ij}(s)|$, whose eigenvalues are bounded by their continuous dependence on matrix entries. Therefore, $M(t)$ is real-valued local martingale with $M(0)=0,\ a.s.$ Based on the semimartingale convergence theorem, following similar steps in the proof of Theorem \ref{AS theorem} leads to the conclusion
\begin{align}\notag
\lim\limits_{t\rightarrow+\infty}x_{i}(t)-\boldsymbol{\theta}=0,\ a.s.,\ i\in\underline{N}.
\end{align}
$\hfill\blacksquare$

\begin{figure}[h]
	\begin{center}
		\includegraphics[width=0.75\linewidth]{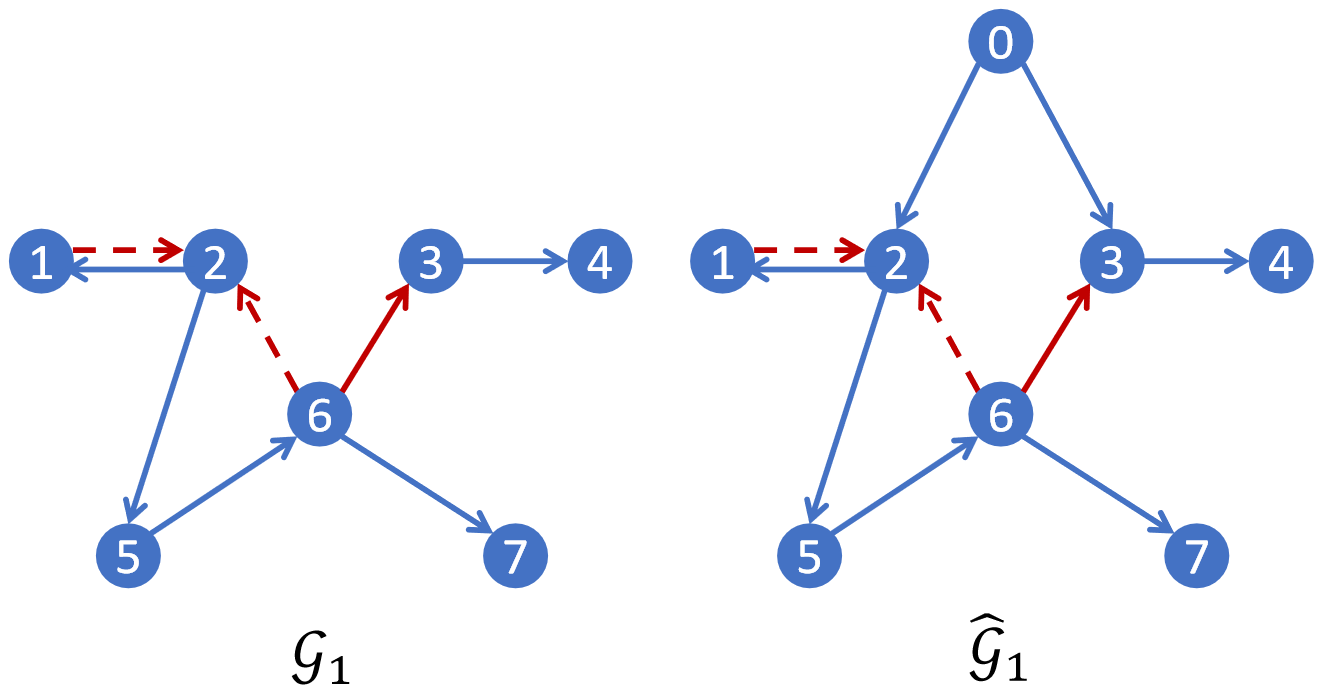}    
		\caption{Directed topology $\mathcal{G}_{1}$ for original FAN \eqref{FAN_noise} and $\widehat{\mathcal{G}}_{1}$ for the corresponding SAN \eqref{SAN_noise} under Lemma \ref{Basic Lemma}. The blue and red solid (dashed) lines represent positive and negative (semi-) definite edges, respectively.}
		\label{G1}                        
	\end{center}
\end{figure}

\begin{figure}[h]
	\begin{center}
		\includegraphics[width=0.90\linewidth]{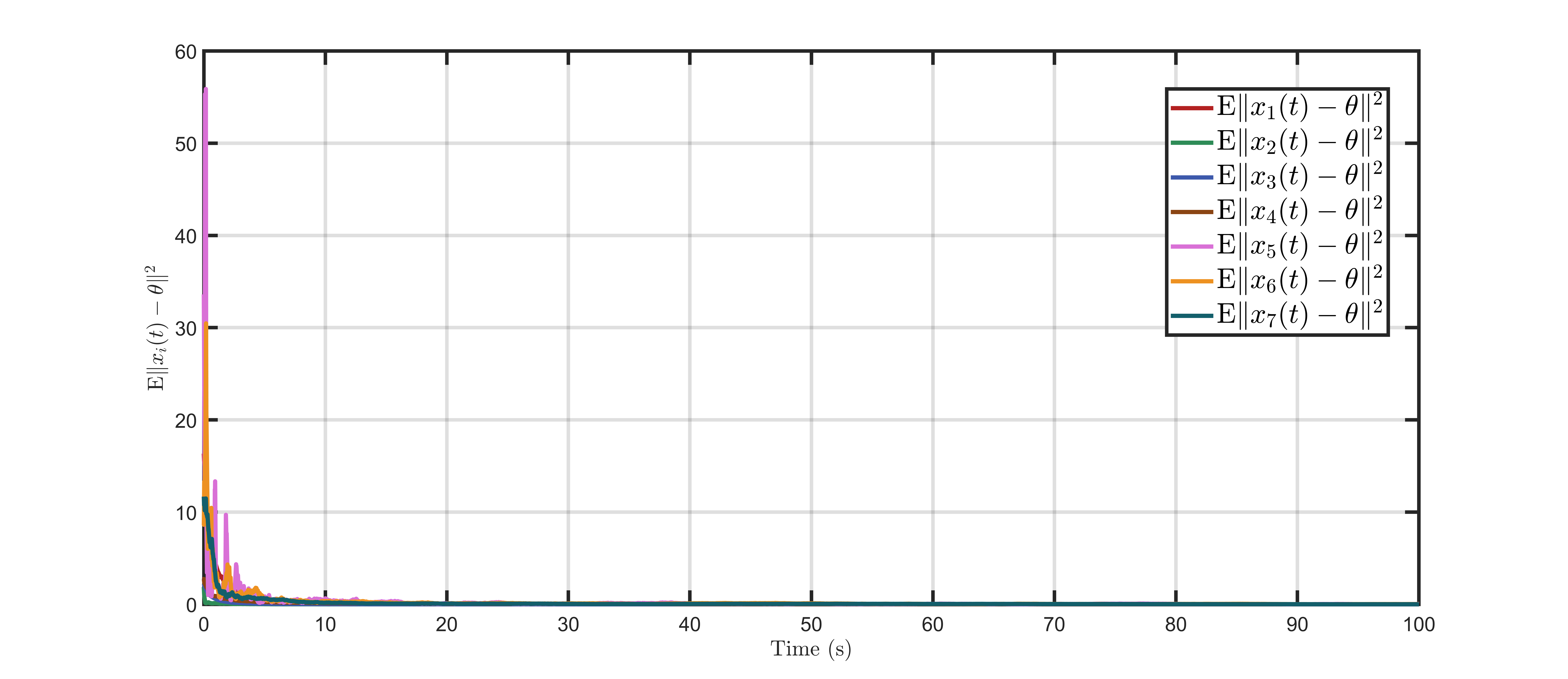}    
		\caption{Mean square non-trivial consensus errors of directed signed matrix-weighted network with fixed topology $\mathcal{G}_{1},\ \widehat{\mathcal{G}}_{1}$ in Figure \ref{G1}.}
		\label{fixed_E}                          
	\end{center}
\end{figure}

\section{Numerical examples}\label{Numerical examples}
\subsection{fixed topology}\label{simulation_fixed subsection}

First, Simulations under original signed matrix-weighted FAN \eqref{FAN_noise} with fixed topology $\mathcal{G}(t)\equiv\mathcal{G}_{1}$ in Figure \ref{G1} are presented. Specifically, the edge matrix weights are
\begin{small}
\begin{equation}\label{G2_weights}\notag
\begin{aligned}
&A_{12}(\mathcal{G}_{1})=\begin{bmatrix}1&0&0\\0&1&0\\0&0&2\end{bmatrix}\succ\boldsymbol{0},\ 
A_{21}(\mathcal{G}_{1})=-\begin{bmatrix}1&0&0\\0&1&0\\0&0&0\end{bmatrix}\preceq\boldsymbol{0},\\
&A_{65}(\mathcal{G}_{1})=\begin{bmatrix}7&1&1\\1&6&0\\1&0&5\end{bmatrix}\succ\boldsymbol{0},\ 
A_{52}(\mathcal{G}_{1})=\begin{bmatrix}10&1&2\\1&8&1\\2&1&12\end{bmatrix}\succ\boldsymbol{0},\\
&A_{76}(\mathcal{G}_{1})=\begin{bmatrix}3.1&0&0\\0&3.1&0\\0&0&3.2\end{bmatrix}\succ\boldsymbol{0},\ 
A_{26}(\mathcal{G}_{1})=-\begin{bmatrix}0&0&0\\0&0&0\\0&0&1\end{bmatrix}\preceq\boldsymbol{0},\\
\end{aligned}
\end{equation}
\end{small}
\noindent and $A_{36}(\mathcal{G}_{1})=-0.1\cdot A_{12}(\mathcal{G}_{1}),\ A_{43}(\mathcal{G}_{1})=A_{12}(\mathcal{G}_{1})$. $\mathcal{G}_{1}$ satisfies Assumption \ref{directed Assumption} with vertices decomposition $\mathcal{V}(\mathcal{G}_{1})=\mathcal{V}_{1}(\mathcal{G}_{1})\cup\mathcal{V}_{2}(\mathcal{G}_{1}),\ \mathcal{V}_{1}(\mathcal{G}_{1})=\{v_{2},v_{3}\},\ \mathcal{V}_{2}(\mathcal{G}_{1})=\{v_{1},v_{4},v_{5},v_{6},v_{7}\}$. Take the desired non-trivial consensus state $\boldsymbol{\theta}=[1,2,-1]^{\top}$. By Lemma \ref{Basic Lemma}, the lower bound $C$ of the coupling coefficient $\delta$ is calculated as $C=7.1440$, and we take $\delta=C+0.1=7.2440$, then external signal $x_{0}=(1+2/\delta)\boldsymbol{\theta}=[1.2761,2.5522,-1.2761]^{\top}$. The informed agents set $\mathcal{V}_{\mathcal{I}}=\{v_{2},v_{3}\}$, and the coupling matrix weights are
\begin{align}\notag
B_{2}(\mathcal{G}_{1})=\left|A_{21}(\mathcal{G}_{1})\right|+\left|A_{26}(\mathcal{G}_{1})\right|,\ B_{3}(\mathcal{G}_{1})=\left|A_{36}(\mathcal{G}_{1})\right|.
\end{align}
The above control strategy is reflected in $\widehat{\mathcal{G}}_{1}$ in Figure \ref{G1}, in which vertex $0$ represents external signal $x_{0}$, and edges $(v_{0},v_{2}),\ (v_{0},v_{3})$ stand for couplings from $x_{0}$ to the informed agents $x_{2}$ and $x_{3}$. Following the above settings, based on Lemma \ref{Basic Lemma}, $-L_{B}(\mathcal{G}_{1})$ is Hurwitz. Set $\sigma_{ji}=0.4$ and $f_{ji}(x)=0.3x$ for $i,j\in\underline{N}$, take control gain $c(t)=\dfrac{1}{1+t}$, then condition \eqref{ct in th1} holds, and Theorem \ref{ms theorem} gurantees the mean square non-trivial consensus. To simulate such behavior, we consider the mean square non-trivial consensus errors $\{\mathbb{E}\Vert x_{i}(t)-\boldsymbol{\theta}\Vert^{2}\}_{i\in\underline{7}}$. By applying the Euler–Maruyama method with a step size of $0.001$, the stochastic differential equations (SDEs) are solved. We generate $1\times10^{3}$ sample paths and take the mean square average, then we obtain Figure \ref{fixed_E}, which shows that the agents achieve mean square non-trivial consensus.

Further, inheriting the above settings, in which control gain $c(t)=\dfrac{1}{1+t}$ satisfies condition \eqref{ct in th2}, Theorem \ref{AS theorem} gives the almost sure non-trivial consensus result, which is depicted by the states evolution in Figure \ref{fixed_AS_simulations}.

\subsection{time-varying topology}\label{simulation_varying}
Next, numerical examples of time-varying topology case are presented. To facilitate simulations, we take the topology $\mathcal{G}(t)$ to be time-varying in a switching pattern: set the time sequence $\{t_{k}\}_{k\in\mathbb{N}}$ as $t_{0}=0$, and $\Delta t_{k}=0.02$ for any $k\in\mathbb{N}$. The switching pattern is for any $l\in\mathbb{N}$, $\mathcal{G}(t)=\mathcal{G}_{2},\ t\in[t_{5l},t_{5l+2})$, $\mathcal{G}(t)=\mathcal{G}_{1},\ t\in[t_{5l+2},t_{5l+3})$, $\mathcal{G}(t)=\mathcal{G}_{3},\ t\in[t_{5l+3},t_{5(l+1)})$, in which topologies $\mathcal{G}_{2},\ \mathcal{G}_{3}$ come from Figure \ref{switching G2 and G3}. The edge matrix weights are given in \eqref{G2_weights} and \eqref{G3_weights}.

\begin{small}
\begin{equation}\label{G2_weights}
\begin{aligned}
&A_{12}(\mathcal{G}_{2})=-\begin{bmatrix}5&2&1\\2&4&1\\1&1&3\end{bmatrix}\prec\boldsymbol{0},\ 
A_{47}(\mathcal{G}_{2})=-\begin{bmatrix}0&0&0\\0&1&0\\0&0&0\end{bmatrix}\preceq\boldsymbol{0},\\
&A_{65}(\mathcal{G}_{2})=\begin{bmatrix}6&1&-1\\1&8&2\\-1&2&6\end{bmatrix}\succ\boldsymbol{0},\ 
A_{67}(\mathcal{G}_{2})=\begin{bmatrix}4&0&0\\0&3&0\\0&0&3\end{bmatrix}\succ\boldsymbol{0},\\ 
&A_{76}(\mathcal{G}_{2})=\begin{bmatrix}5&1&0\\1&6&-1\\0&-1&5\end{bmatrix}\succ\boldsymbol{0},\ 
A_{43}(\mathcal{G}_{2})=-\begin{bmatrix}2&0&1\\0&0&0\\1&0&3\end{bmatrix}\preceq\boldsymbol{0},
\end{aligned}
\end{equation}
\begin{equation}\label{G3_weights}
\begin{aligned}
&A_{34}(\mathcal{G}_{3})=\begin{bmatrix}1&0&0\\0&2&0\\0&0&0\end{bmatrix}\succeq\boldsymbol{0},
A_{37}(\mathcal{G}_{3})=-\begin{bmatrix}2&-1&0\\-1&2&0\\0&0&1\end{bmatrix}\prec\boldsymbol{0},\\ 
&A_{43}(\mathcal{G}_{3})=\begin{bmatrix}2&-1&0\\-1&3&1\\0&1&4\end{bmatrix}\succ\boldsymbol{0},\ 
A_{75}(\mathcal{G}_{3})=\begin{bmatrix}2&0&0\\0&3&0\\0&0&1\end{bmatrix}\succ\boldsymbol{0},\\
&A_{52}(\mathcal{G}_{3})=\begin{bmatrix}8&2&1\\2&8&1\\1&1&6\end{bmatrix}\succ\boldsymbol{0},\ 
A_{21}(\mathcal{G}_{3})=-\begin{bmatrix}3&1&-1\\1&5&2\\-1&2&5\end{bmatrix}\prec\boldsymbol{0},
\end{aligned}
\end{equation}
\end{small}
\noindent and $A_{51}(\mathcal{G}_{2})=A_{65}(\mathcal{G}_{2}),\ A_{26}(\mathcal{G}_{2})=A_{36}(\mathcal{G}_{2})=-0.5I_{3},\ A_{62}(\mathcal{G}_{2})=A_{47}(\mathcal{G}_{2}),\ A_{63}(\mathcal{G}_{3})=A_{15}(\mathcal{G}_{3})=-I_{3}\prec\boldsymbol{0}$.

Under Lemma \ref{Basic Lemma}'s technique, for the desired non-trivial consensus state $\theta=[1,2,-1]^{\top}$, related parameters are set as follow:\\
for $t\in[t_{5l},t_{5l+2})$, let
\begin{equation}\notag
\begin{aligned}
&\delta(t)=7.0495,\ x_{0}(t)=[1.2837, 2.5674, -1.2837]^{\top},\\ &B_{1}(t)=|A_{12}(\mathcal{G}_{2})|\succ\boldsymbol{0},\ 
B_{2}(t)=|A_{26}(\mathcal{G}_{2})|\succ\boldsymbol{0},\ 
B_{3}(t)=|A_{36}(\mathcal{G}_{2})|\succ\boldsymbol{0},\\
&B_{4}(t)=|A_{43}(\mathcal{G}_{2})|+|A_{47}(\mathcal{G}_{2})|\succ\boldsymbol{0},\ 
B_{6}(t)=|A_{62}(\mathcal{G}_{2})|\succeq\boldsymbol{0},
\end{aligned}
\end{equation}
for $t\in[t_{5l+3},t_{5(l+1)})$, let
\begin{equation}\notag
\begin{aligned}
&\delta(t)=3.1000,\ x_{0}(t)=[1.6452, 3.2903, -1.6452]^{\top},\\ &B_{1}(t)=\left|A_{15}(\mathcal{G}_{3})\right|\succ\boldsymbol{0},\ B_{2}(t)=\left|A_{21}(\mathcal{G}_{3})\right|\succ\boldsymbol{0},\\ &B_{3}(t)=\left|A_{37}(\mathcal{G}_{3})\right|\succ\boldsymbol{0},\ 
B_{6}(t)=\left|A_{63}(\mathcal{G}_{3})\right|\succ\boldsymbol{0},
\end{aligned}
\end{equation}
and for $t\in[t_{5l+2},t_{5l+3})$, the parameters are set the same as in Subsection \ref{simulation_fixed subsection}. Same as in Subsection \ref{simulation_fixed subsection}, $\widehat{\mathcal{G}}_{2}$ and $\widehat{\mathcal{G}}_{3}$ in Figure \ref{switching G2 and G3} depict the above control strategy. Following the above settings, it can be verified that matrix $L_{B}(\mathcal{G}_{1})+L_{B}^{\top}(\mathcal{G}_{1}),\ L_{B}(\mathcal{G}_{2})+L_{B}^{\top}(\mathcal{G}_{2})$ and $L_{B}(\mathcal{G}_{3})+L_{B}^{\top}(\mathcal{G}_{3})$ are all positive definite. Set $\sigma_{ji}=0.4$ and $f_{ji}(x)=0.3x$ for $i,j\in\underline{N}$, take control gain $c(t)=\dfrac{1}{1+t}$, then conditions \eqref{ct in th3} holds, and Theorem \ref{switching ms theorem} gurantees the mean square non-trivial consensus. To simulate such behavior, we consider the mean square non-trivial consensus errors $\{\mathbb{E}\Vert x_{i}(t)-\boldsymbol{\theta}\Vert^{2}\}_{i\in\underline{7}}$. By applying the Euler–Maruyama method with a step size of $0.001$, the stochastic differential equations (SDEs) are solved. We generate $1\times10^{3}$ sample paths and take the mean square average, then we obtain Figure \ref{switch_E}, which shows that the agents achieve mean square non-trivial consensus.

Further, with the above settings, $c(t)=\dfrac{1}{1+t}$ satisfies conditions \eqref{ct in th4}, therefore, Theorem \ref{switching AS theorem} gurantees the almost sure non-trivial consensus result, which is depicted by the states evolution in Figure \ref{switch_AS_simulations}.

Remark \ref{ct Remark} under Theorem \ref{AS theorem} discusses the idea for designing control gain $c(t)$ in our study. The significance of these design conditions can also be reflected by simulation results in Figure \ref{ct=(1+t)^-2} and Figure \ref{ct=(1+t)^-0.33}, in which we choose $c(t)=\dfrac{1}{(1+t)^2}$ and $c(t)=\dfrac{1}{(1+t)^\frac{1}{3}}$, respectively that do not satisfy conditions \eqref{ct in th4}, and keep other settings remain unchanged. It is shown that, the almost sure non-trivial consensus result can not be guranteed, due to the vanishing cumulative effect or infinite energy of noise-induced effects.

\section{Conclusions}\label{Conclusions}
This paper studies non-trivial consensus\textemdash a relatively novel convergence result that has remained largely unexplored\textemdash on directed signed matrix-weighted networks with additive and multiplicative measurement noises. For any prespecified non-trivial consensus value, the proposed protocols guarantee that all agents converge towards this shared non-zero value through the synthesis of the informed agents selection, the coupling items and external signals determination and the control gain design. The protocols require mild connectivity conditions and impose no requirements on structural (un)balance properties. The results are first obtained in the mean square and then in the almost sure senses. Both the fixed and the time-varying topology cases are investigated. The numerical simulations indicate the effectiveness and correctness of the theoretical results. Outlook for future work includes fixed-time or prescribed-time non-trivial consensus control on MASs with actuator faults or input saturation.

\section*{Acknowledgments}
This work was supported in part by the Major Research Plan of the National Natural Science Foundation of China (Grant No. 92267101, 62573301, 62503335), and in part by the Startup Grant of Shenzhen University.



\begin{thebibliography}{00}


\bibitem{consensus problems}
	R. Olfati-Saber, R. M. Murray, Consensus problems in networks of agents with switching topology and time-delays, IEEE Trans. Autom. Control 49 (9) (2004) 1520-1533.

	
\bibitem{consensus_IS_2023}
	J. Lv, C. Wang, B. Liu, Y. Kao, Y. Jiang, 
	Fully distributed prescribed-time consensus control of multiagent systems under fixed and switching topologies, Inf. Sci. 648 (2023) 119538.

\bibitem{SDE_Mao}
	X. Mao, Stochastic differential equations and applications, 2nd ed.
	Chichester, U.K.: Horwood Publishing (2007).


\bibitem{noises_SIAM2018}
	X. Zong, T. Li, J.-F. Zhang, Consensus conditions of continuous-time multi-agent systems with additive and multiplicative measurement noises, SIAM J. Control Optimiz. 56 (1) (2018) 19-52.
	
\bibitem{noises_delay_Auto2019}	
	X. Zong, T. Li, J.-F. Zhang, Consensus conditions of continuous-time multi-agent systems with time-delays and measurement noises, Automatica 99 (2019) 412–419.
	
\bibitem{noises_Markov_TAC2019}
	A. Jadbabaie, A. Olshevsky, Scaling laws for consensus protocols subject to noise, IEEE Trans. Autom. Control 64 (4) (2019) 1389-1402.
	
\bibitem{MS_IF2023}
	F. Sun, C. Lu, W. Zhu, J. Kurths, Data-sampled mean-square consensus of hybrid multi-agent systems with time-varying delay and multiplicative noises, Inf. Sci. 624 (2023) 674-685.	

	
\bibitem{fault_noises_IF2023}
	C. Liu, J. Zhao, B. Jiang, R. J. Patton, Fault-tolerant consensus control of multi-agent systems under actuator/sensor faults and channel noises: A distributed anti-attack strategy, Inf. Sci. 623 (2023) 1–19.

\bibitem{noises_etc_Auto2026}
	R. Jia, Y.-H. Ni, G. Wang, Event-triggered global and group consensus for linear multiagent systems with multiplicative noises, Automatica 183 (2026) 112595.	
	

\bibitem{antagonistic}
	C. Altafini, Consensus problems on networks with antagonistic interactions, IEEE Trans. Autom. Control 58 (4) (2013) 935-946.
	
\bibitem{opinion separation}
	W. Xia, M. Cao, K. H. Johansson, Structural balance and opinion separation in trust–mistrust social networks, IEEE Trans. Control Netw. Syst. 3 (1) (2016) 46-56.
	
\bibitem{IB consensus}
	D. Meng, M. Du, Y. Jia, Interval bipartite consensus of networked agents associated with signed digraphs, IEEE Trans. Autom. Control 61 (12) (2016) 3755-3770.
	
\bibitem{bipartite_IF2024}
	W. Zhang, Q. Huang, A. Alhudhaif, Event-triggered fixed-time bipartite consensus for nonlinear disturbed multi-agent systems with leader-follower and leaderless controller, Inf. Sci. 662 (2024) 120243.
	
\bibitem{noises_signed_2019}
	J. Hu, Y. Wu, T. Li, B. K. Ghosh, Consensus control of general linear multiagent systems with antagonistic interactions and communication noises, IEEE Trans. Autom. Control 64 (5) (2019) 2122-2127.	
	
\bibitem{noises_private_TAC2024}
	J. Wang, J. Ke, J.-F. Zhang, Differentially private bipartite consensus over signed networks with time-varying noises, IEEE Trans. Autom. Control 69 (9) (2024) 5788-5803.


\bibitem{small oscillations}
	S. E. Tuna, Synchronization of small oscillations, Automatica 107 (2019) 154-161.


\bibitem{LC}
	S. E. Tuna, Synchronization under matrix-weighted Laplacian, Automatica 73 (2016) 76-81.


\bibitem{network science}
	N. E. Friedkin, A. V. Proskurnikov, R. Tempo, S. E. Parsegov, Network science on belief system dynamics under logic constraints, Science 354 (6310) (2016) 321-326.


\bibitem{multiple interdependent topics}
	M. Ye, M. Trinh, Y.-H. Lim, B. Anderson, H.-S. Ahn, Continuous-time opinion dynamics on multiple interdependent topics, Automatica 115 (2020) 108884.


\bibitem{Luan}
	Y. Luan, X. Wu, J. L\"{u}, Coevolutionary dynamics of multidimensional opinions over coopetitive influence networks, Automatica 177 (2025) 112279.


\bibitem{Graph effective resistance}
	P. Barooah, J. P. Hespanha, Graph effective resistance and distributed control: Spectral properties and applications, Proc. 45th IEEE Conf. Decis. Control (2006) 3479-3485.


\bibitem{Barooah Estimation}
	P. Barooah, J. P. Hespanha, Estimation from relative measurements: Electrical analogy and large graphs, IEEE Trans. Signal Process 56 (6) (2008) 2181-2193.
	

\bibitem{positive spanning tree}
	M. H. Trinh, C. V. Nguyen, Y.-H. Lim, H.-S. Ahn, Matrix-weighted consensus and its applications, Automatica 89 (2018) 415-419.


\bibitem{balancing set}
	C. Wang, L. Pan, H. Shao, D. Li, Y. Xi, Characterizing bipartite consensus on signed matrix-weighted networks via balancing set, Automatica 141 (2022) 110237.


\bibitem{leader-following matrix-weighted consensus}
	M. H. Trinh, M. Ye, H.-S. Ahn, B. D. O. Anderson, Matrix-weighted consensus with leader-following topologies, Proc. IEEE 11th Asian Control Conf. (2017) 1795-1800.


\bibitem{Su TCASII matrix weighted bipartite consensus}
	H. Su, J. Chen, Y. Yang, Z. Rong, The bipartite consensus for multi-agent systems with matrix-weight-based signed network, IEEE Trans. Circuits Syst. II-Express Briefs 67 (10) (2020) 2019-2023.


\bibitem{Su TAC directed matrix weighted consensus}
	H. Su, S. Miao, Consensus on directed matrix-weighted networks, IEEE Trans. Autom. Control 68 (4) (2023) 6711-6726.


\bibitem{Pan TCASII bipartite consensus}
	L. Pan, H. Shao, M. Mesbahi, Y. Xi, D. Li, Bipartite consensus on matrix-valued weighted networks, IEEE Trans. Circuits Syst. II-Express Briefs 66 (8) (2019) 1441-1445.
	
\bibitem{Privacy Preservation_noises_matrix-weight_TCNS2025}
	P. Wang, H. Shao, L. Pan , Member, W. Yan, N. Li, Computationally light privacy preservation of matrix-weighted average consensus, IEEE Trans. Control Netw. Syst. 12 (2) (2025) 1651-1661.

\bibitem{noises_mw_TCNS2020}
	D. R. Foight, M. H. de Badyn, M. Mesbah, Performance and design of consensus on matrix-weighted and time-scaled graphs, IEEE Trans. Control Netw. Syst. 7 (4) (2020) 1812-1822.

\bibitem{non-trivial consensus}
	A. Proskurnikov, A. Matveev, M. Cao, Consensus and polarization in Altafini’s model with bidirectional time-varying network topologies, in Proc. 53rd IEEE Conf. Decis. Control (2014) 2112-2117.


\bibitem{UAS}
	E. Atkins, A. Ollero, A. Tsourdos, Unmanned Aircraft Systems. Hoboken: John Wiley \& Sons Ltd (2016).


\bibitem{Consensus manipulation}
	X. Chen, H. Liang, Y. Zhang, Y. Wu, Consensus manipulation in social network group decision making with value-based opinion evolution, Inf. Sci. 647 (2023) 119441.


\bibitem{Manipulating opinions}
	P. Bolzern, A. Colombo, C. Piccardi, Manipulating opinions in social networks with community structure, IEEE Trans. Netw. Sci. Eng. 11 (1) (2024) 185-196.


\bibitem{manipulation Automatica}
	F. Liu, S. Cui, G. Chen, W. Mei, H. Gao, Modeling, analysis, and manipulation of co-evolution between appraisal dynamics and opinion dynamics, Automatica 167 (2024) 111797.


\bibitem{Valcher}
	M. E. Valcher, P. Misra, On the consensus and bipartite consensus in high-order multi-agent dynamical systems with antagonistic interactions, Syst. Control Lett. 66 (2014) 94-103.


\bibitem{SCIS Non-trivial consensus}
	T. Niu, B. Mao, L. Wang, X. Wu, Non-trivial consensus control on directed signed networks, Sci. China-Inf. Sci. 69 (3) (2025) 132203.
	
	
\bibitem{my_matrix-weighted_NTC}
	T. Niu, B. Mao, X. Wu, T. Huang, Non-trivial consensus on directed matrix-weighted networks with cooperative and antagonistic interactions, https://arxiv.org/abs/2602.11822.
	

\bibitem{switching matrix-weighted consensus}
	L. Pan, H. Shao, M. Mesbahi, Y. Xi, D. Li, Consensus on matrix-weighted switching networks, IEEE Trans. Autom. Control 66 (12) (2021) 5990-5996.


\bibitem{switching matrix-weighted cluster consensus}
	L. Pan, H. Shao, M. Mesbahi, D. Li, Y. Xi, Cluster consensus on matrix-weighted switching networks, Automatica 141 (2022) 110308.
	
\bibitem{noises_signed_TSMC2022}
	Y. Du, Y. Wang, Z. Zuo, Mean square bipartite consensus for multiagent systems with antagonistic information and time-varying topologies, IEEE Trans. Syst. Man Cybern. Syst. 52 (3) (2022) 1744-1754.
	
\bibitem{FAN&SAN}
	H. Shao, L. Pan, M. Mesbahi, Y. Xi, D. Li, Distributed neighbor selection in multiagent networks, IEEE Trans. Autom. Control 68 (11) (2023) 6711-6726.

\bibitem{naive}
	W. Xia, M. Cao, Clustering in diffusively coupled networks, Automatica 47 (2011) 2395-2405.	

	
\bibitem{semimartingale}
	R. S. Lipster, A. N. Shiryayev, Theory of Martingales, Kluwer Acad. Publ., Dordrecht, The Netherlands (1989).
	
	
\end{thebibliography}



\begin{figure}[H]
	\centering
	\begin{minipage}[b]{0.98\linewidth}
		\centering
		\includegraphics[width=\linewidth]{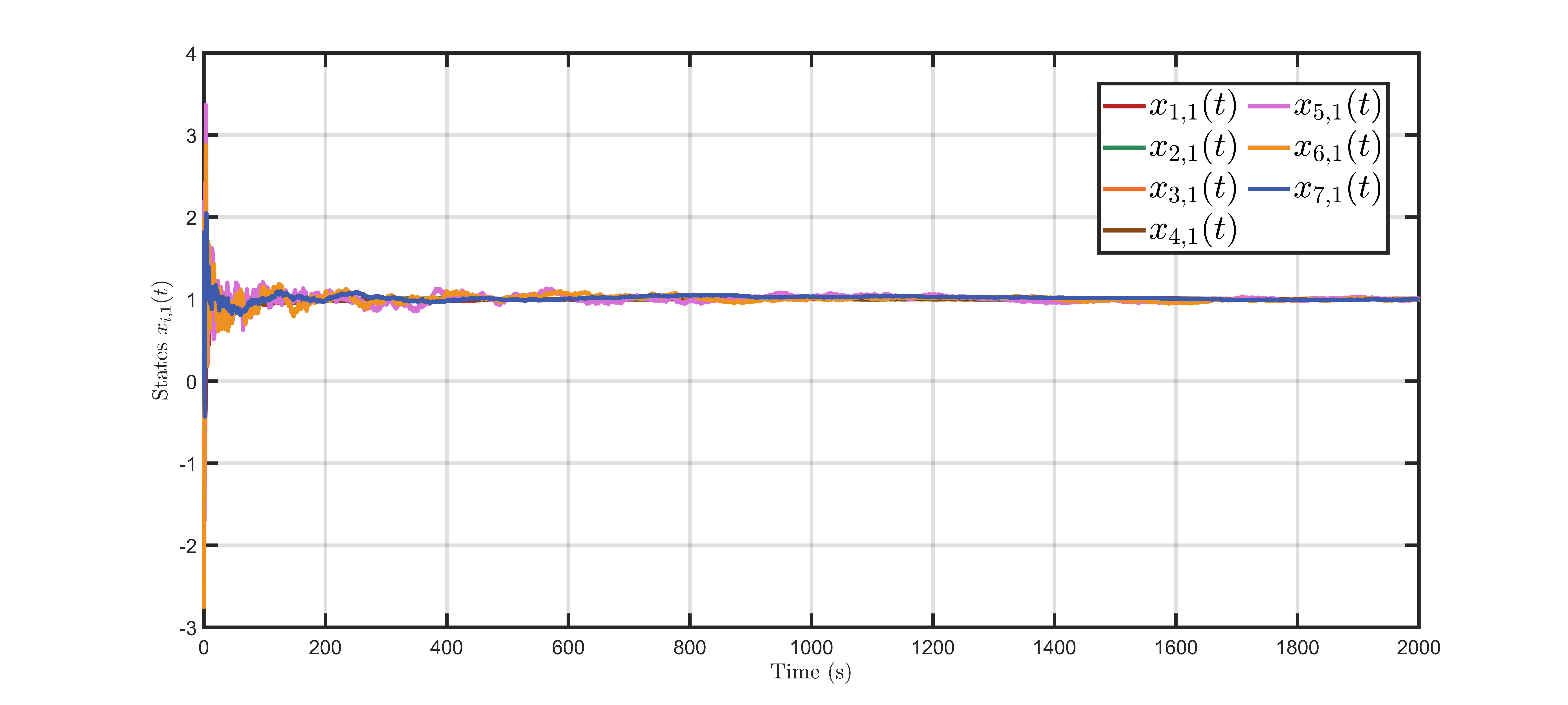}
	\end{minipage}
	\vfill
	\begin{minipage}[b]{0.98\linewidth}
		\centering
		\includegraphics[width=\linewidth]{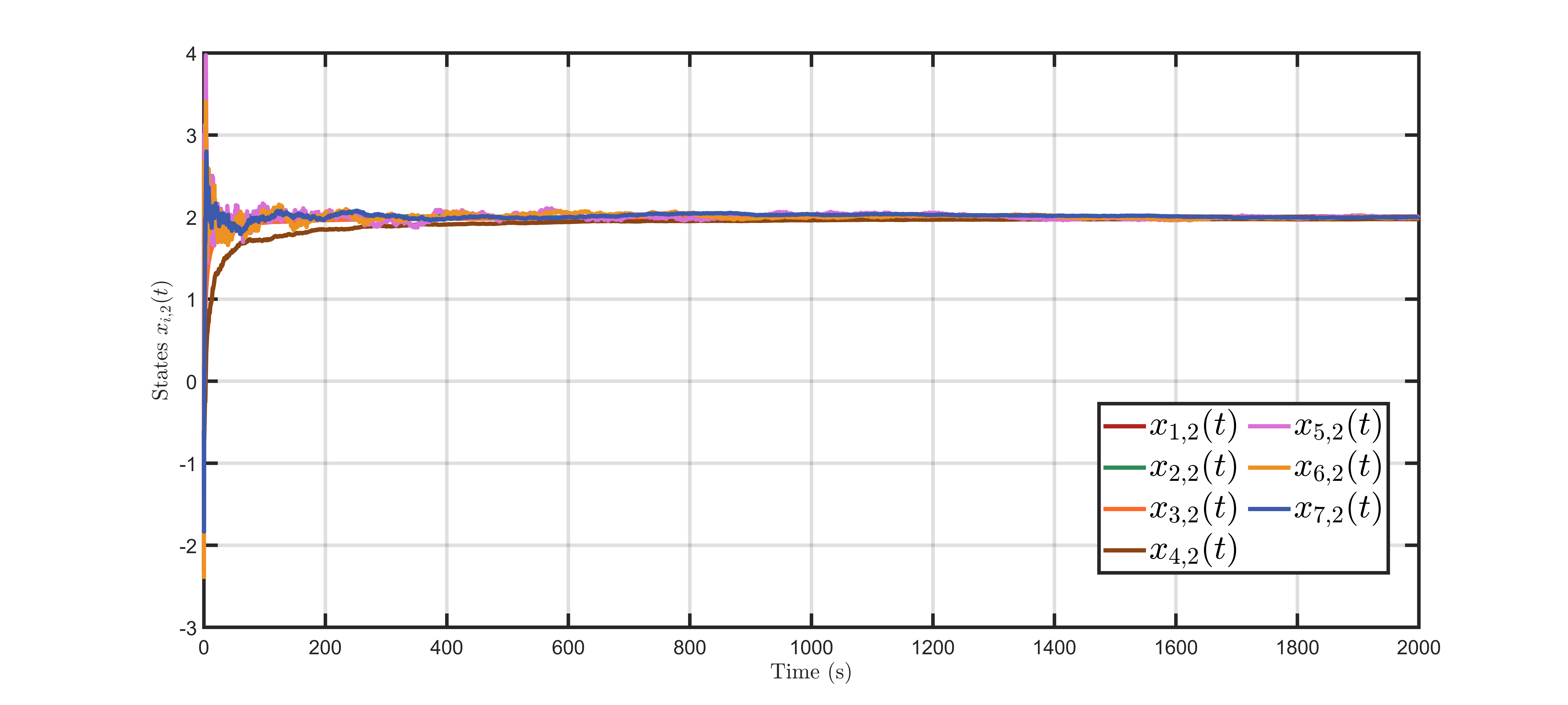}
	\end{minipage}
	\vfill
	\begin{minipage}[b]{0.98\linewidth}
		\centering
		\includegraphics[width=\linewidth]{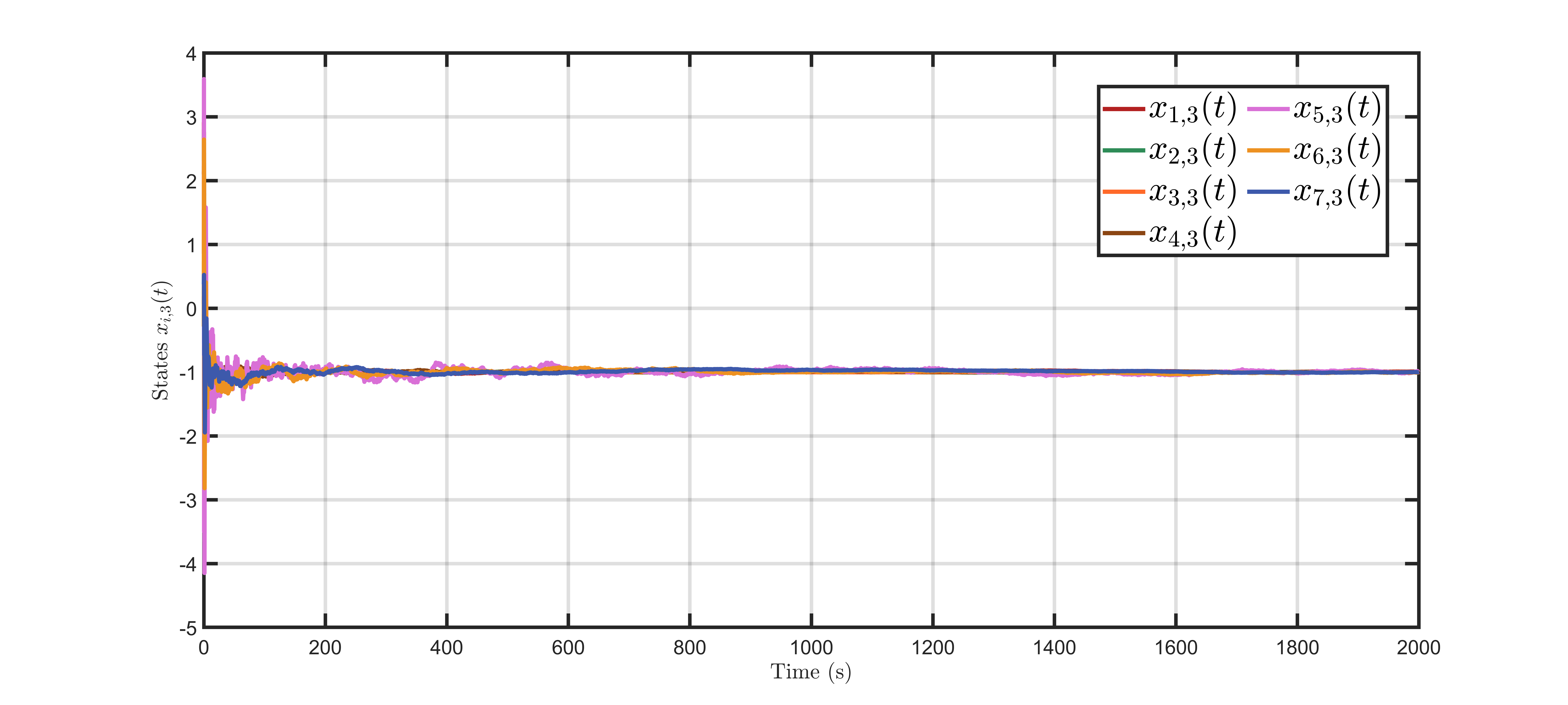}
	\end{minipage}
	\caption{States evolution of directed matrix-weighted network with fixed topology $\widehat{\mathcal{G}}_{1}$ in Figure \ref{G1}.}
	\label{fixed_AS_simulations}
\end{figure}

\begin{figure}[H]
	\centering
	\begin{minipage}[b]{0.75\linewidth}
		\centering
		\includegraphics[width=\linewidth]{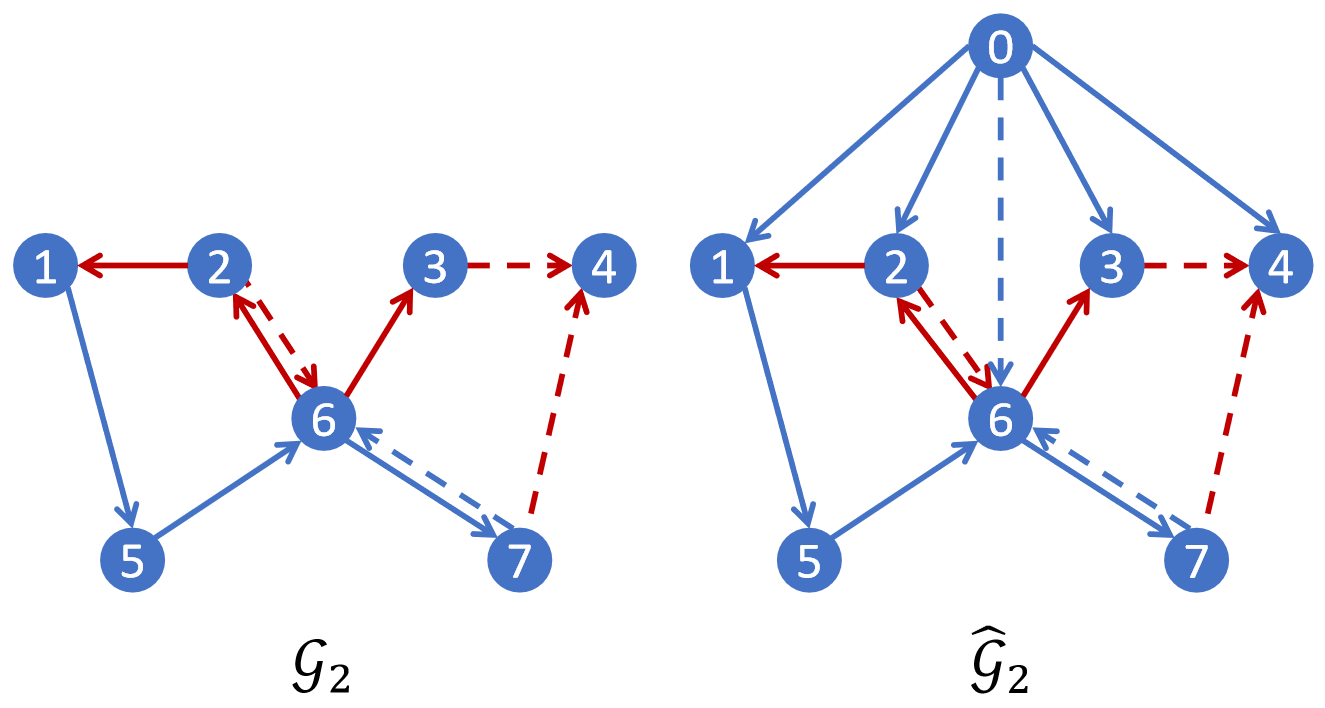}
	\end{minipage}
	\vfill
	\begin{minipage}[b]{0.75\linewidth}
		\centering
		\includegraphics[width=\linewidth]{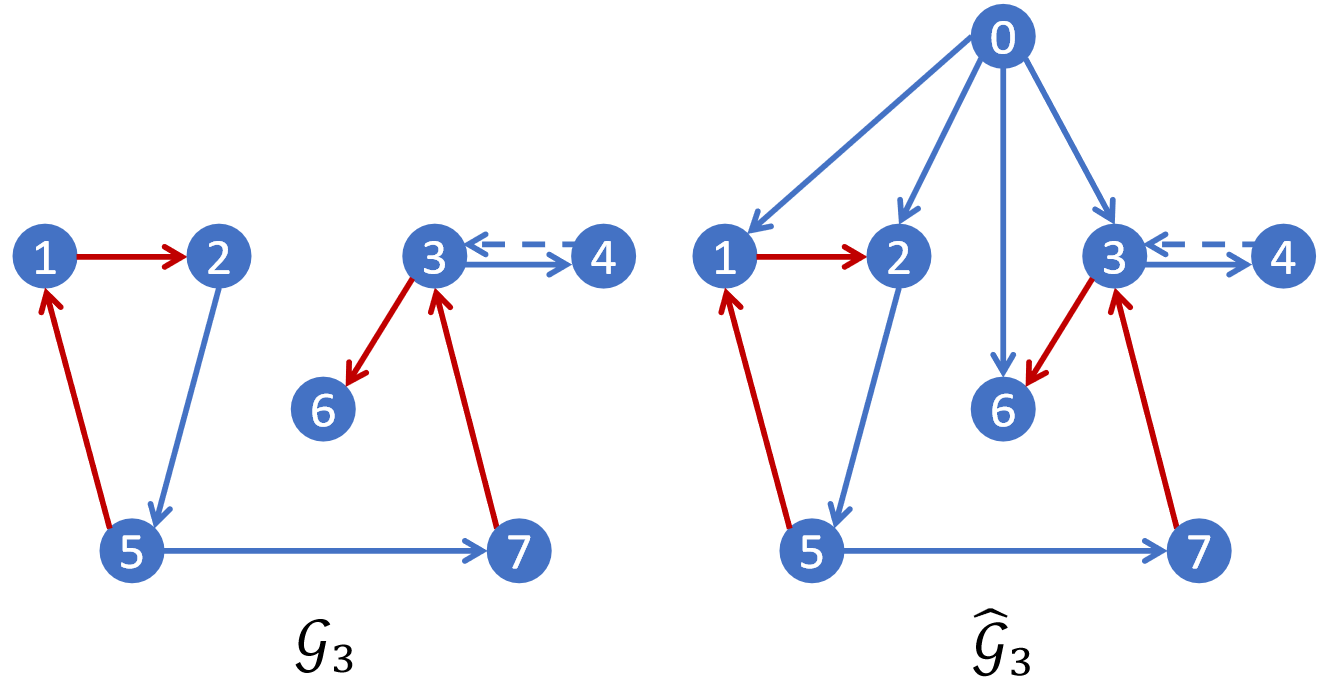}
	\end{minipage}
	\caption{Topology $\mathcal{G}_{2},\ \widehat{\mathcal{G}}_{2}$ and $\mathcal{G}_{3},\ \widehat{\mathcal{G}}_{3}$ for switching network \eqref{FAN_noise} \eqref{SAN_noise} under Lemma \ref{Basic Lemma}. The blue and red solid (dashed) lines represent positive and negative (semi-) definite edges, respectively.}
	\label{switching G2 and G3}
\end{figure}

\begin{figure}[H]
	\begin{center}
		\includegraphics[width=0.90\linewidth]{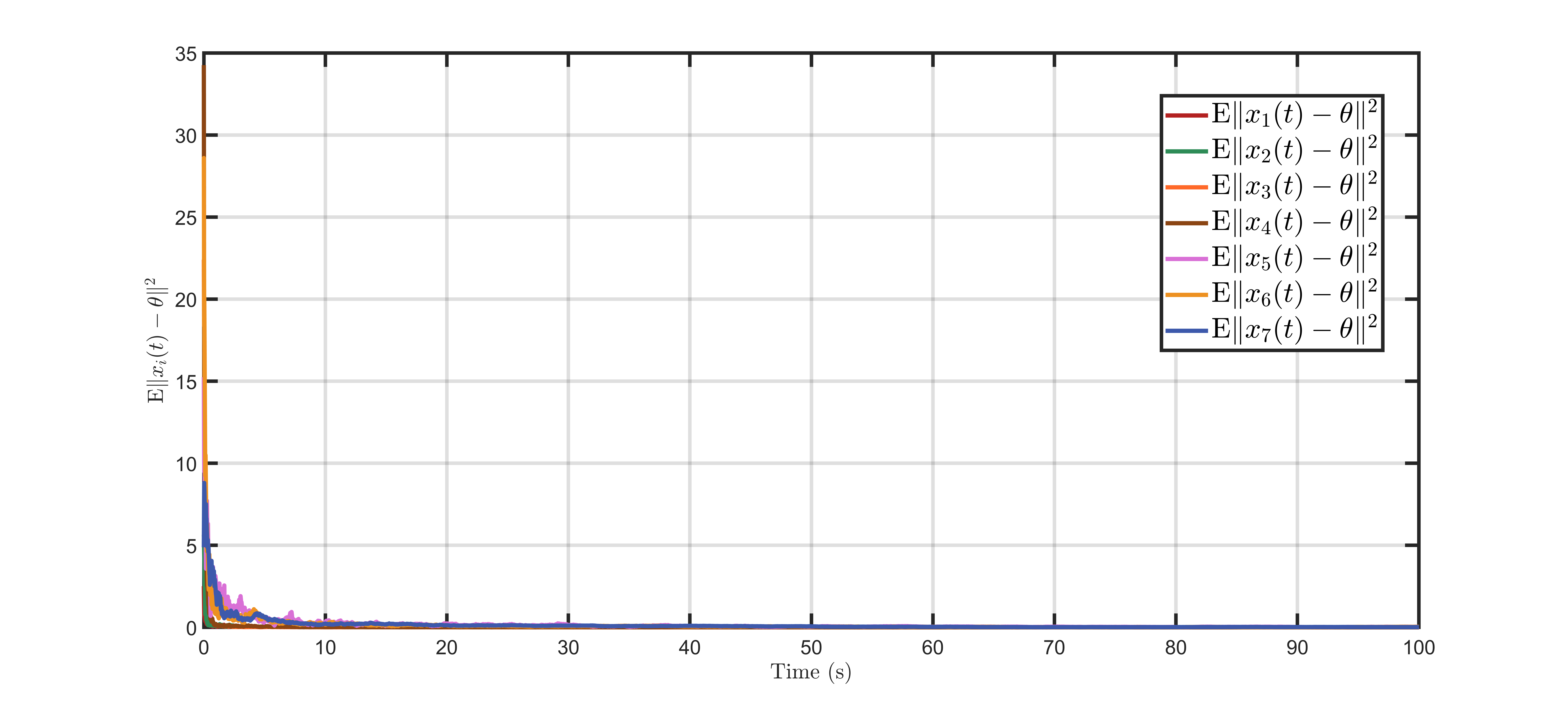}    
		\caption{Mean square non-trivial consensus error of directed matrix-weighted network with time-varying topology $\widehat{\mathcal{G}}(t)$ switching among $\widehat{\mathcal{G}}_{1}$ in Figure \ref{G1} and $\widehat{\mathcal{G}}_{2}$, $\widehat{\mathcal{G}}_{3}$ in Figure \ref{switching G2 and G3}.}
		\label{switch_E}                            
	\end{center}
\end{figure}

\begin{figure}[H]
	\centering
	\begin{minipage}[b]{0.98\linewidth}
		\centering
		\includegraphics[width=\linewidth]{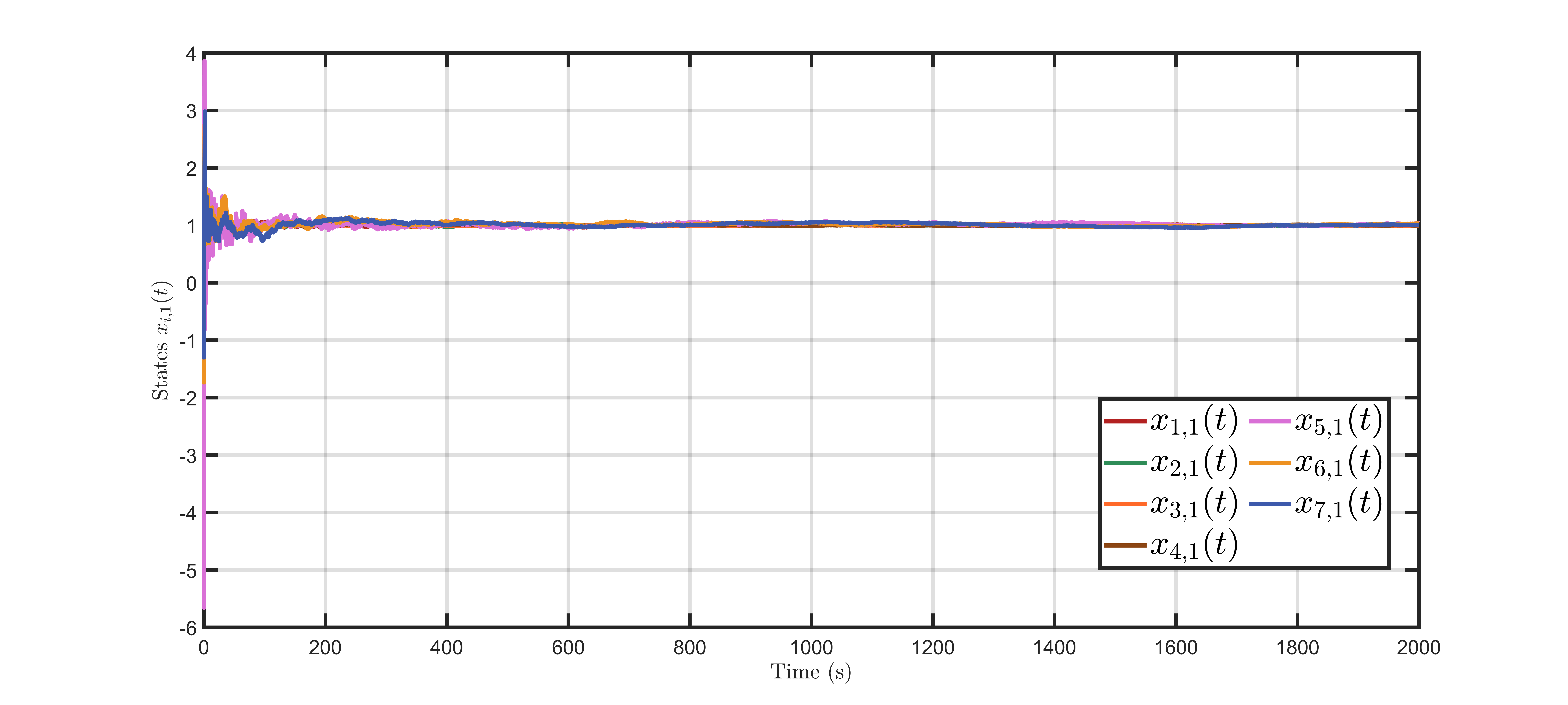}
	\end{minipage}
	\vfill
	\begin{minipage}[b]{0.98\linewidth}
		\centering
		\includegraphics[width=\linewidth]{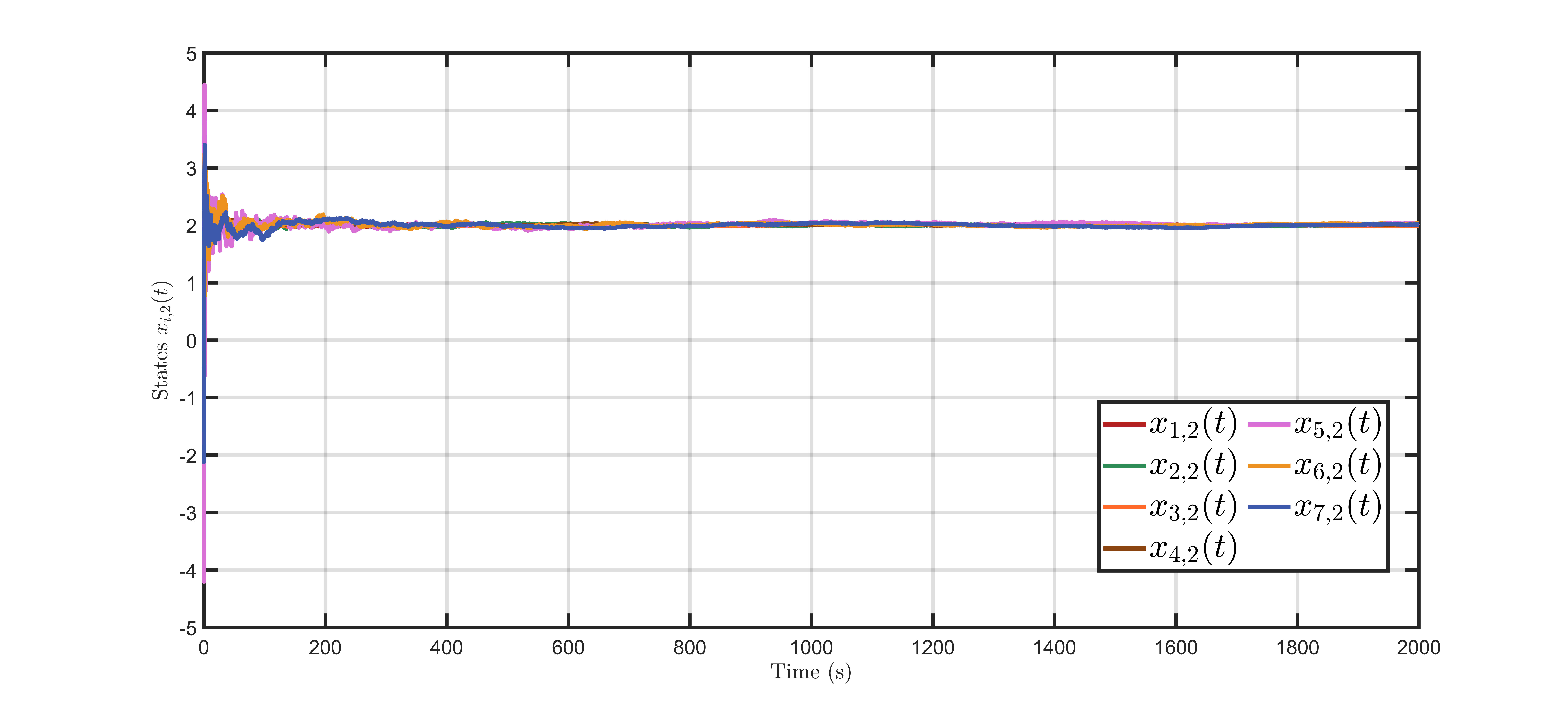}
	\end{minipage}
	\vfill
	\begin{minipage}[b]{0.98\linewidth}
		\centering
		\includegraphics[width=\linewidth]{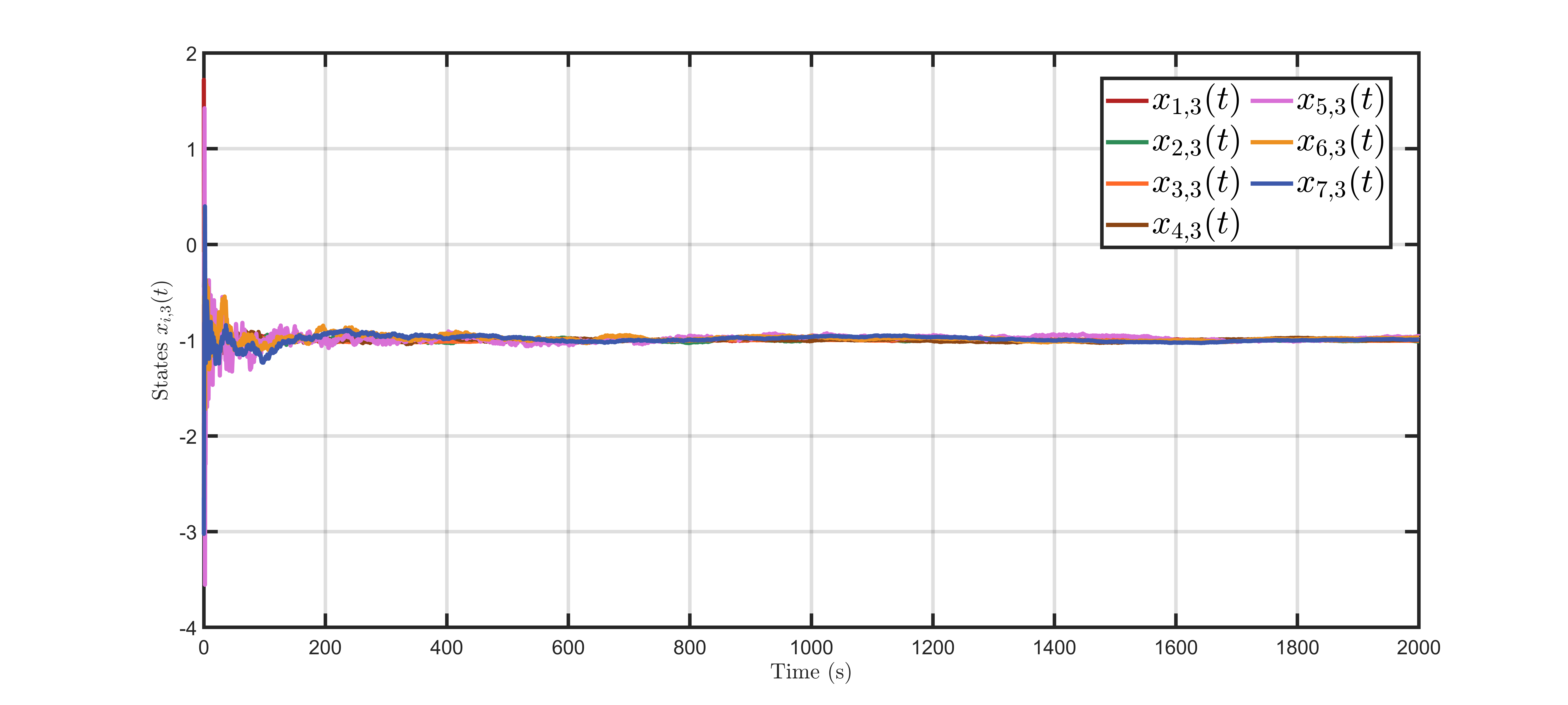}
	\end{minipage}
	\caption{States evolution of directed matrix-weighted network with time-varying topology $\widehat{\mathcal{G}}(t)$ switching among $\widehat{\mathcal{G}}_{1}$ in Figure \ref{G1} and $\widehat{\mathcal{G}}_{2}$, $\widehat{\mathcal{G}}_{3}$ Figure \ref{switching G2 and G3}.}
	\label{switch_AS_simulations}
\end{figure}

\begin{figure}[H]
	\centering
	\begin{minipage}[b]{0.98\linewidth}
		\centering
		\includegraphics[width=\linewidth]{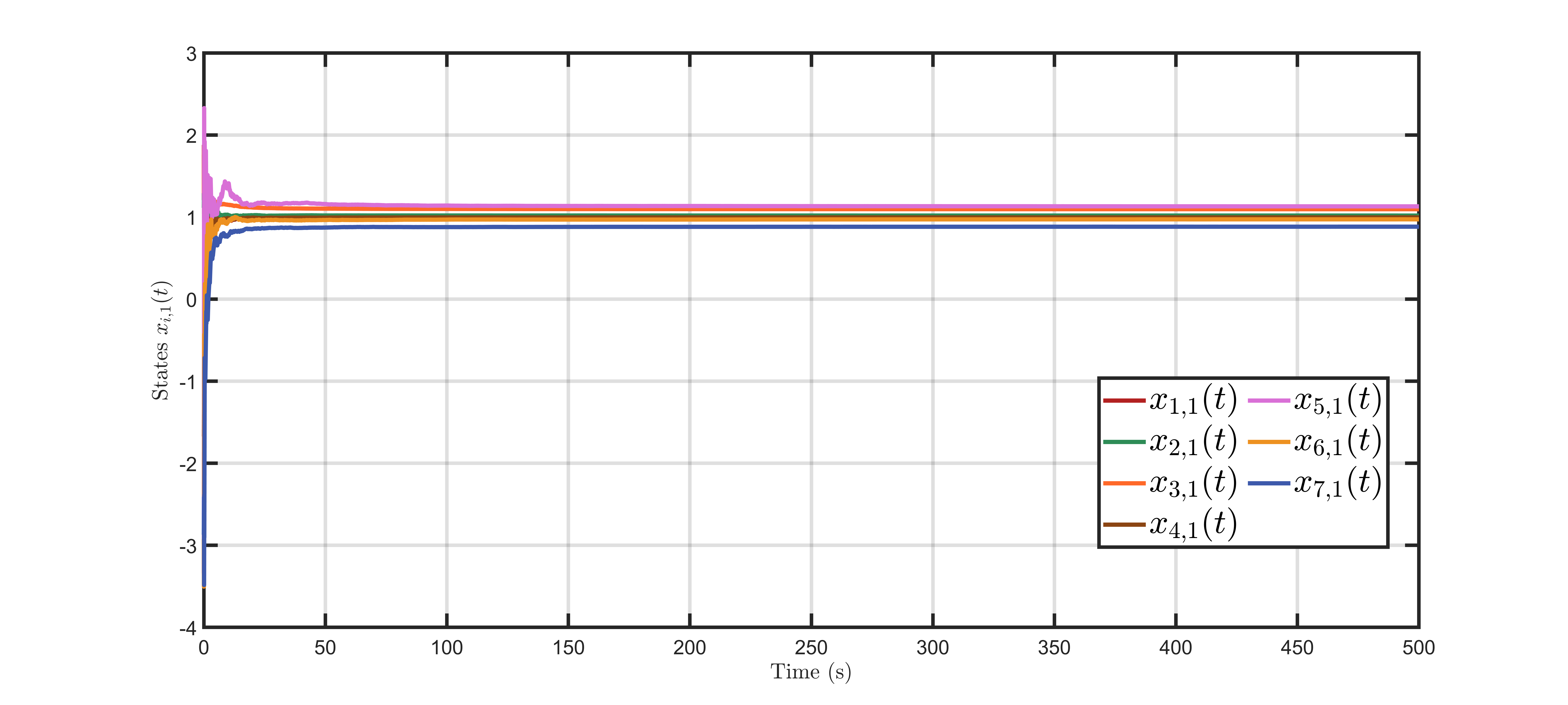}
	\end{minipage}
	\vfill
	\begin{minipage}[b]{0.98\linewidth}
		\centering
		\includegraphics[width=\linewidth]{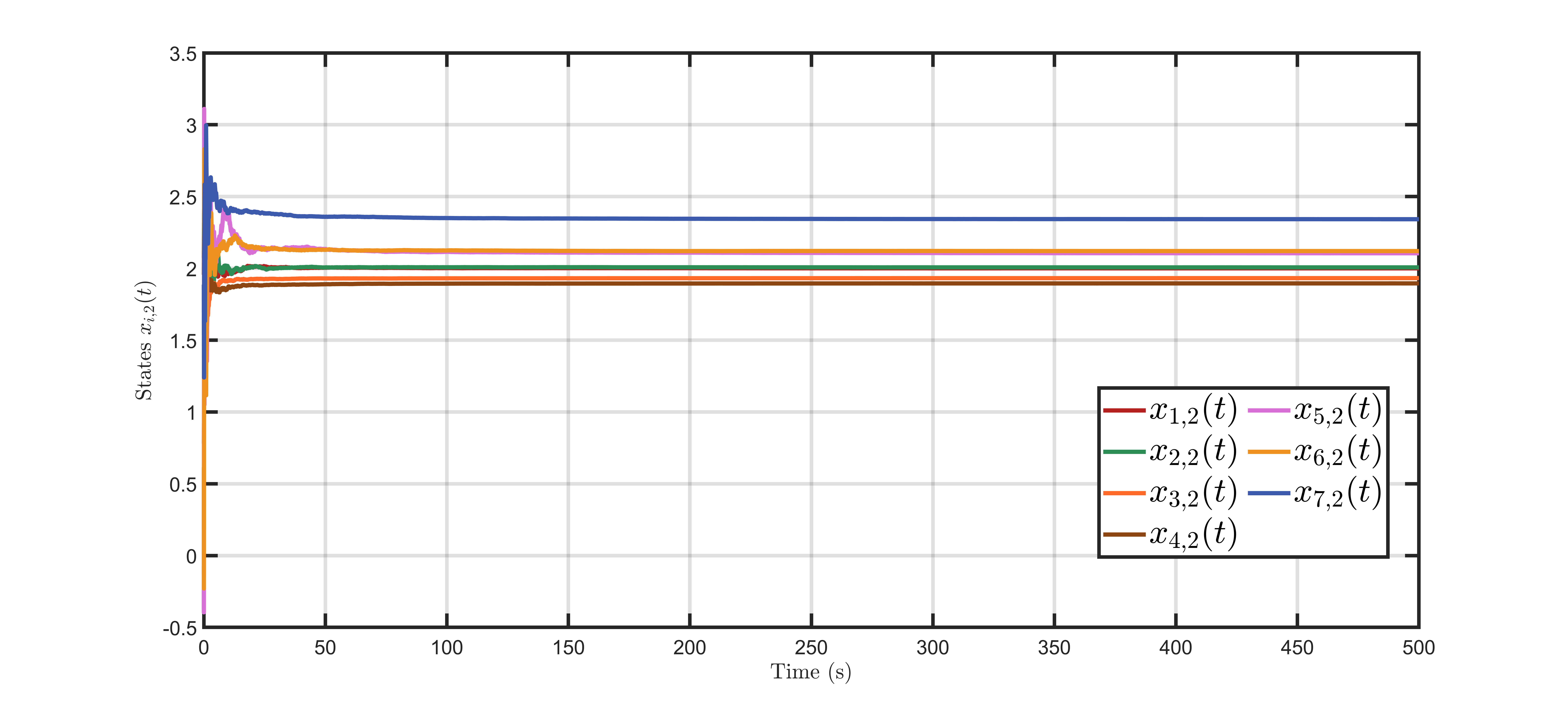}
	\end{minipage}
	\vfill
	\begin{minipage}[b]{0.98\linewidth}
		\centering
		\includegraphics[width=\linewidth]{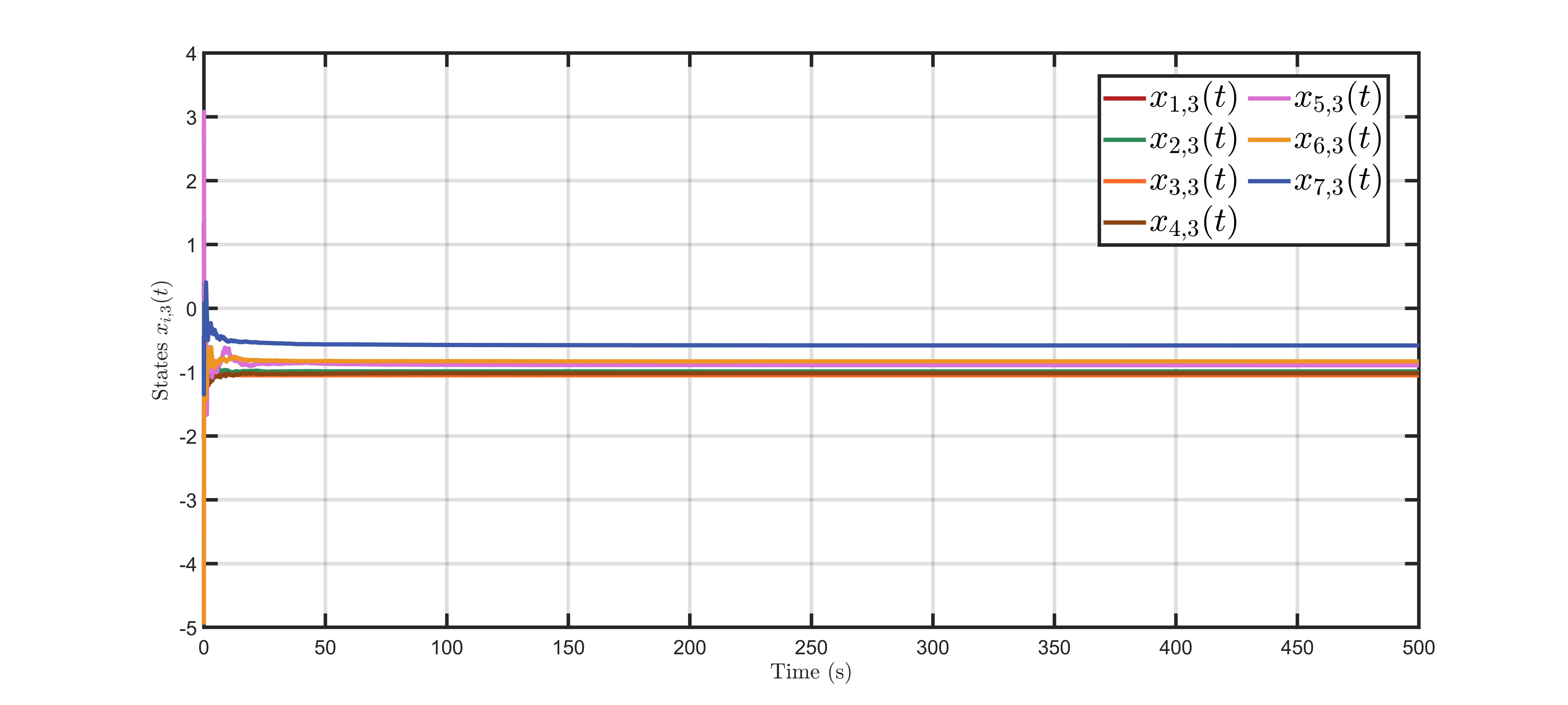}
	\end{minipage}
	\caption{States evolution of directed matrix-weighted network with time-varying topology $\widehat{\mathcal{G}}(t)$ switching among $\widehat{\mathcal{G}}_{1}$ in Figure \ref{G1} and $\widehat{\mathcal{G}}_{2}$, $\widehat{\mathcal{G}}_{3}$ Figure \ref{switching G2 and G3}. The control gain $c(t)$ is chosen as $(1+t)^{-2}$.}
	\label{ct=(1+t)^-2}
\end{figure}

\begin{figure}[H]
	\centering
	\begin{minipage}[b]{0.98\linewidth}
		\centering
		\includegraphics[width=\linewidth]{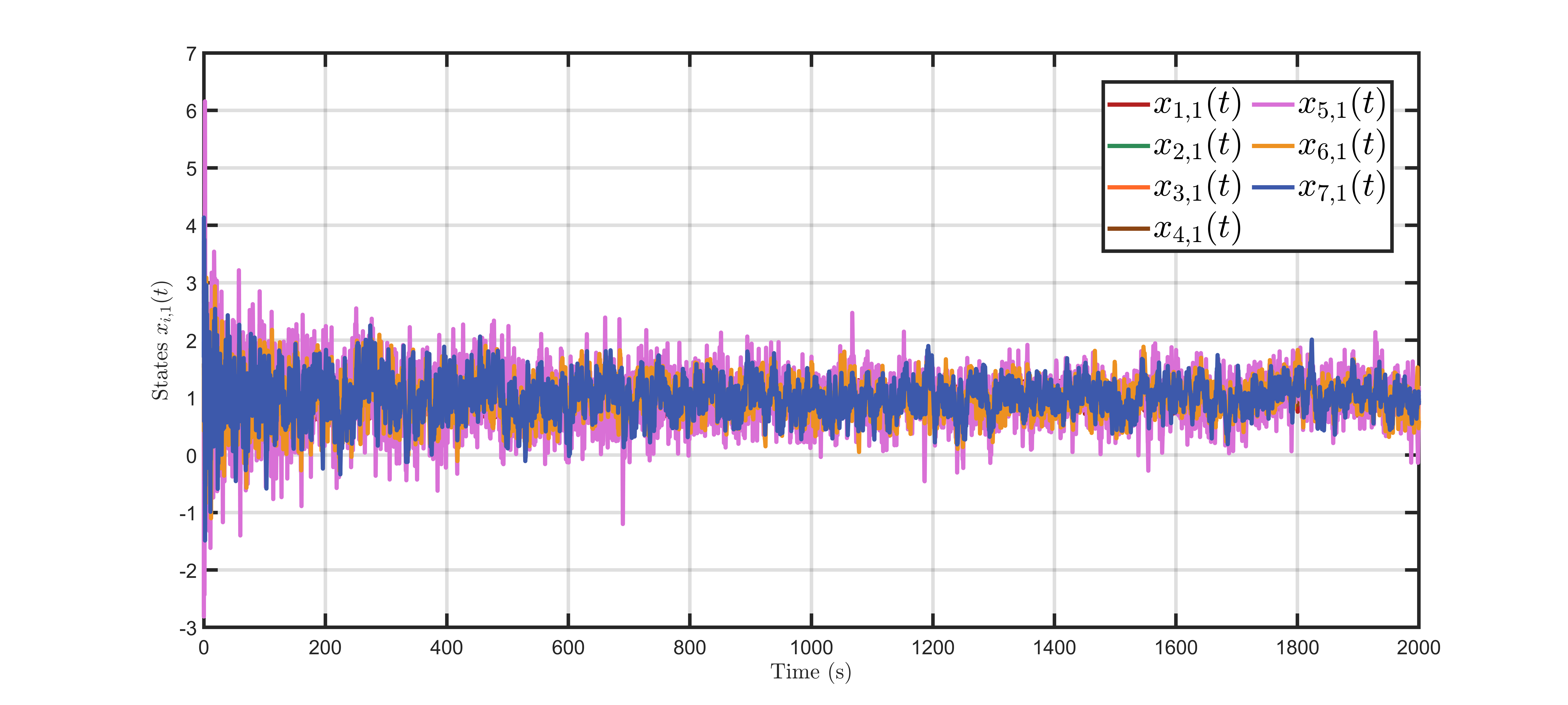}
	\end{minipage}
	\vfill
	\begin{minipage}[b]{0.98\linewidth}
		\centering
		\includegraphics[width=\linewidth]{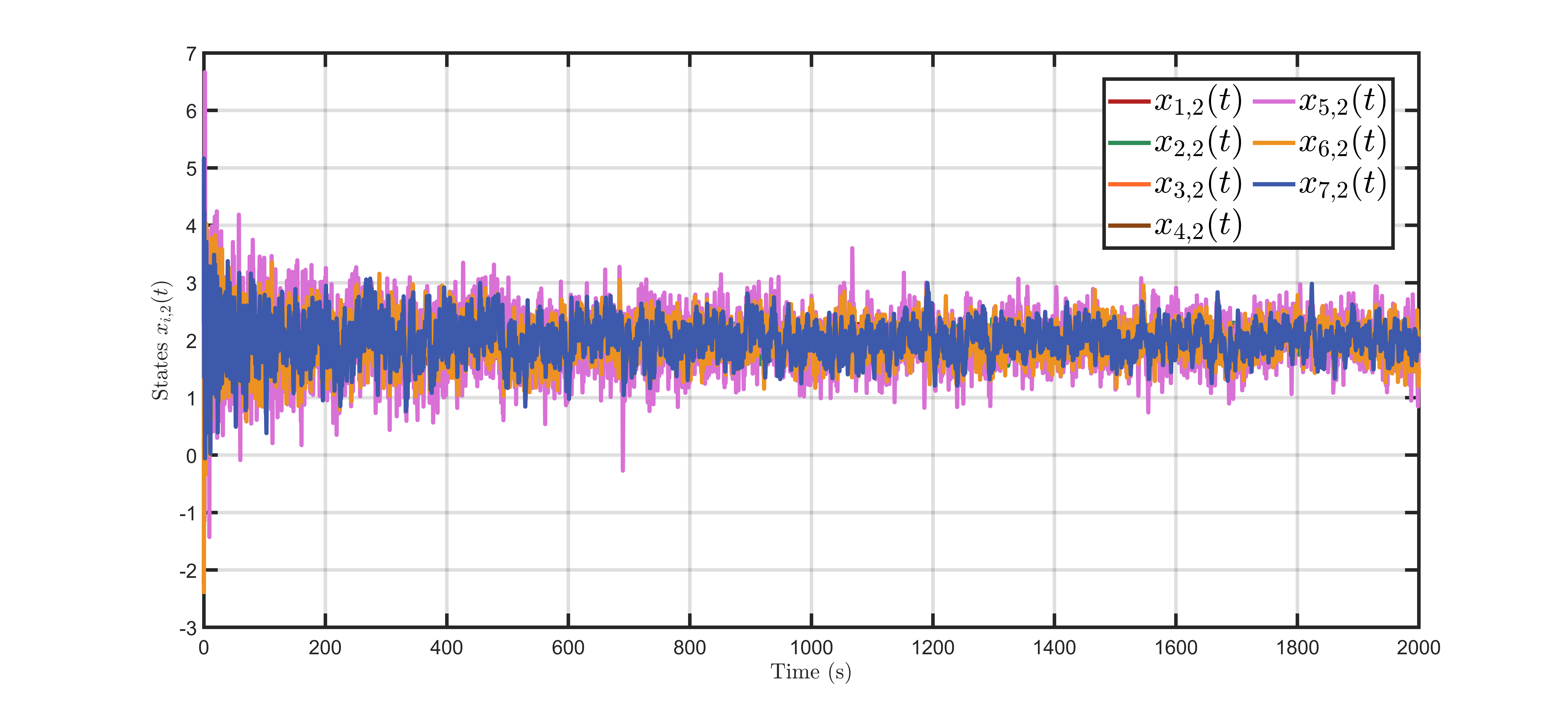}
	\end{minipage}
	\vfill
	\begin{minipage}[b]{0.98\linewidth}
		\centering
		\includegraphics[width=\linewidth]{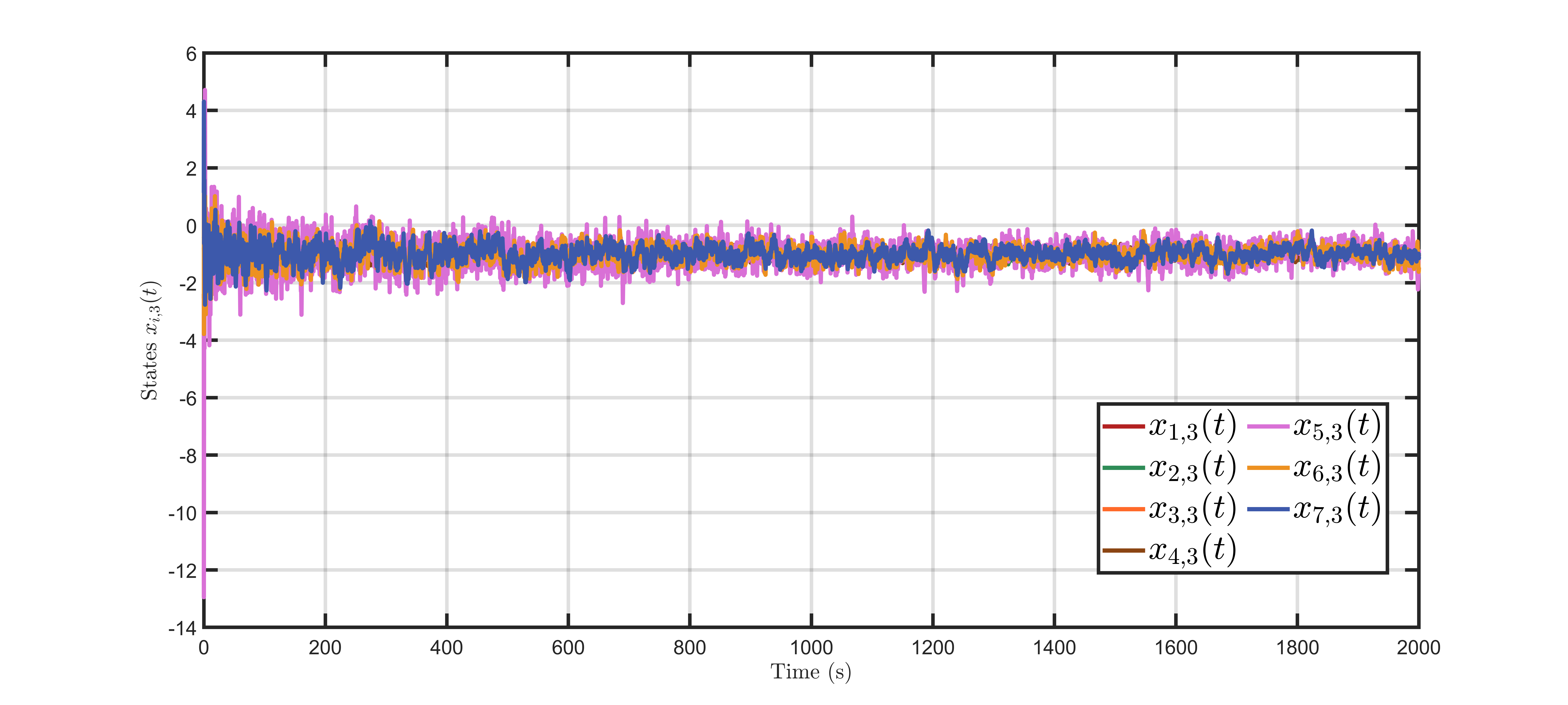}
	\end{minipage}
	\caption{States evolution of directed matrix-weighted network with time-varying topology $\widehat{\mathcal{G}}(t)$ switching among $\widehat{\mathcal{G}}_{1}$ in Figure \ref{G1} and $\widehat{\mathcal{G}}_{2}$, $\widehat{\mathcal{G}}_{3}$ Figure \ref{switching G2 and G3}. The control gain $c(t)$ is chosen as $(1+t)^{-\frac{1}{3}}$.}
	\label{ct=(1+t)^-0.33}
\end{figure}

\end{document}